\documentclass[11pt]{article}

\usepackage[utf8]{inputenc}
\usepackage[T1]{fontenc}
\usepackage[a4paper,margin=1in]{geometry}
\usepackage{amsmath,amssymb,amsthm,mathtools}
\usepackage{microtype}
\usepackage{enumitem}
\usepackage{hyperref}
\usepackage{booktabs}
\usepackage{cleveref,aliascnt}
\usepackage{xspace}
\usepackage{thm-restate,color,xcolor}
\usepackage{amsmath,amssymb,bbm,amsthm}
\usepackage{fullpage}
\usepackage{boxedminipage}
\usepackage{amsmath}
\usepackage{mathtools}
\usepackage{framed}
\usepackage[framemethod=tikz]{mdframed}
\usepackage{titlesec}
\usepackage{lipsum}%
\usepackage{cleveref,aliascnt}
\usepackage{tikz}
\usepackage{multirow}
\usepackage{comment}

\usepackage{wrapfig}
\usepackage[ruled,vlined,linesnumbered]{algorithm2e}
\usepackage{mathtools}
\DeclarePairedDelimiter\ceil{\lceil}{\rceil}

\crefname{algocf}{algorithm}{algorithms}
\Crefname{algocf}{Algorithm}{Algorithms}

\newtheorem{theorem}{Theorem}[section]
\newtheorem{lemma}[theorem]{Lemma}
\newtheorem{claim}[theorem]{Claim}
\newtheorem{definition}[theorem]{Definition}
\newtheorem{remark}[theorem]{Remark}
\newtheorem{corollary}[theorem]{Corollary}
\newtheorem{proposition}[theorem]{Proposition}
\newtheorem{observation}[theorem]{Observation}

\newcommand{\PiPath}{P_i}
\newcommand{\V}{V}
\newcommand{\E}{E}
\newcommand{\G}{G}

\newcommand{\dist}{\mathrm{dist}}

\newcommand{\distX}[1]{\mathrm{dist}_{#1}}

\newcommand{\Hi}{H_i}

\newcommand{\Gleq}[1]{G_{\le #1}}
\newcommand{\GleqCirc}[1]{G^{\circ}_{\le #1}}

\newcommand{\HleqCirc}[1]{H^{\circ}_{\le #1}}
\newcommand{\card}[1]{\left|#1\right|}

\newcommand{\Ball}{B}

\newcommand{\pp}{\mathcal{P}}

\title{New Greedy Spanners and Applications}
\author{Elizaveta Popova\thanks{Weizmann Institute. Email: \texttt{elizaveta.popova@weizmann.ac.il}.}\and Elad Tzalik\thanks{Weizmann Institute. Email: \texttt{elad.tzalik@weizmann.ac.il}. Supported by the Adams Fellowship Program of the Israel Academy of Sciences and Humanities.} }
\date{
\today
}

\begin{document}
\maketitle

\begin{abstract}
We present a simple greedy procedure to compute an  $(\alpha,\beta)$-spanner for a graph $G$. We then show that this procedure is useful for building fault-tolerant spanners, as well as spanners for weighted graphs.

Our first main result is an algorithm that, given a multigraph $G$, outputs an $f$ edge fault-tolerant $(k,k-1)$-spanner $H$ of size $O(fn^{1+\frac1k})$ which is tight. To our knowledge, this is the first tight result concerning the price of fault tolerance in spanners which are not multiplicative, in any model of faults.

Our second main result is a new construction of a spanner for weighted graphs. We show that any weighted graph $G$ has a subgraph $H$ with $O(n^{1+\frac{1}{k}})$ edges such that any path $P$ of hop-length $\ell$ in $G$ has a replacement path $P'$ in $H$ of weighted length $\leq w(P)+(2k-2)w^{(1/2)}(P)$ where $w(P)$ is the total edge weight of $P$, and $w^{(1/2)}$ denotes the sum of the largest $\lceil \frac{\ell}{2} \rceil$ edge weights along $P$. Moreover, we show such approximation is optimal for shortest paths of \emph{hop-length} $2$. To our knowledge, this is the first construction of a ``spanner'' for weighted graphs that strictly improves upon the stretch of multiplicative $(2k-1)$-spanners for all non-adjacent vertex pairs, while maintaining the same size bound.

Our technique is based on using clustering and ball-growing, which are methods commonly used in \emph{designing} spanner algorithms, to \emph{analyze} simple greedy algorithms. This allows us to combine the flexibility of clustering approaches with the unique properties of the greedy algorithm to get improved bounds. In particular, our methods give a very short proof that the parallel greedy spanner adds $O(kn^{1+\frac{1}{k}})$ edges, improving upon known bounds.

\end{abstract}

\section{Introduction}
\label{sec:intro}
Let $G$ be an $n$-vertex weighted graph.
A \emph{spanner} of $G$ is a subgraph that approximately preserves distances.
Formally, a subgraph $H$ is a $t$-spanner of $G$ if $\dist_H (u,v) \leq t \cdot \dist_G (u,v)$ for all $u,v\in V$ (where $\dist_X$ denotes the shortest path distance in a graph $X$).
Since their introduction by Peleg and Ullman \cite{PelegU88} and Peleg and Sch\"{a}ffer \cite{PelegS89}, spanners have been found to be extremely useful for a wide variety of applications, including network tasks like routing and synchronization \cite{PelegU89a,ThorupZ01,Cowen01,CowenW04,PelegU88}, preconditioning linear systems \cite{ElkinEST08}, distance estimation \cite{ThorupZ05,BaswanaS04},  and many others. The first tight construction of spanners was given by \cite{AlthoferDDJS93}, which proved that any graph admits a $(2k-1)$-spanner with $O(n^{1+\frac{1}{k}})$ edges, which is best possible\footnote{The $O(n^{1+\frac{1}{k}})$ upper bound is tight \emph{assuming the girth conjecture of Erd\H{o}s}. Nevertheless, the algorithm of \cite{AlthoferDDJS93} is optimal \emph{unconditionally}.}.

The foundational work of Althöfer et al. \cite{AlthoferDDJS93}, focuses on providing a stretch $(2k-1)$ for \emph{adjacent} pairs of vertices, and highlights that graphs of girth $>2k$ are a bottleneck for such sparsification. To handle these bottlenecks, much work has focused on $(\alpha,\beta)$-spanners.
\begin{definition}[(\(\alpha,\beta\))-spanner]
Let $G=(V,E)$ be a graph.  
A subgraph $H \subseteq G$ is called an \emph{$(\alpha,\beta)$-spanner} of $G$ if for all $u,v \in V$,
\[
  \dist_H(u,v) \;\leq\; \alpha \cdot \dist_G(u,v) + \beta .
\]
When $\beta = 0$ we call $H$ a \emph{multiplicative spanner}, and when $\alpha = 1$ we call it an \emph{additive spanner}.
\end{definition}

Intuitively, $(\alpha,\beta)$-spanners have effectively multiplicative stretch $\alpha$ for vertices that are far apart in $G$, while $\beta$ masks the local bottlenecks such a graph may have -- in order to achieve an upper bound $O(n^{1+\frac{1}{k}})$  on the size of the spanner. The first indication that one can provide improved stretch
for distant vertices is the seminal work of Elkin and Peleg \cite{ElkinP04}. In \cite{ElkinP04} the authors  constructed $(1+\varepsilon,\beta)$ spanners of size $O_{\varepsilon,k}(n^{1+\frac{1}{k}})$
edges and $\beta = O\left(\frac{\log(k)}{\varepsilon}\right)^{\log k}$ for every integer $k$ and $\varepsilon\in(0,1)$, as well as a $(k-1,O(k))$-spanner.
This sparked many follow-up works to investigate the tradeoffs of size, distance and stretch achievable in graph sparsification and in particular the $(k,k-1)$ spanner of Baswana et al. \cite{BaswanaKMP10}, the spanner of Thorup and Zwick \cite{ThorupZ06}, and many additional works e.g. \cite{Chechik13,Ben-LevyP20,ElkinGN22}.  

In real-world scenarios, spanners are frequently employed in systems whose components are susceptible to occasional breakdowns.
It is therefore desirable to have spanners possessing resilience to such failures, leading to the notion of \emph{fault-tolerant (FT) spanners}.
Two extensively studied types of faults are \emph{edge faults} and \emph{vertex faults}.
Spanners in the edge/vertex fault-tolerant (E/VFT) settings are defined as follows:

\begin{definition}[E/VFT Spanner]
    An $f$-EFT (VFT) $(\alpha,\beta)$-spanner of $G$ is a subgraph $H$ such that for every set $F$ of at most $f$ edges (vertices) in $G$, it holds that $H-F$ is an $(\alpha,\beta)$-spanner of $G-F$.
\end{definition}

E/VFT spanners have received major interest in recent years; see, e.g.,
\cite{ChechikLPR10,DinitzK11,BodwinDPW18,BodwinP19,DinitzR20,BodwinDR21,BodwinDR22,Parter22} and references therein. The current state of affairs is that multiplicative FT spanners are well understood; this is reflected by essentially tight bounds on the size of fault tolerant spanners in case of vertex faults \cite{BodwinP19,Parter22}, edge faults \cite{BodwinDR22}, bounded degree faults \cite{BodwinHP24}, and color faults \cite{PetruschkaST24_ITCS,ParterPST25}. 

In contrast, the right size bounds required from an $f$-FT $(\alpha,\beta)$ spanners is poorly understood, and no tight bounds are known for given $(\alpha,\beta)$ that improve the stretch obtained by the multiplicative $f$-FT $(2k-1)$-spanner \emph{in any FT model}. The work of \cite{BraunschvigCPS15} studied FT additive and $(\alpha,\beta)$-spanners and constructed $f$-EFT $(k+\varepsilon,k-1)$-spanners with $O \left(f\cdot \left(\frac{(k-1)}{\varepsilon} \right)^{2f} n^{1+\frac{1}{k}} \right)$ many edges. An alternative construction may be given using the method of Dinitz and Krauthgamer \cite{DinitzK11} to obtain $f$-EFT $(k,k-1)$-spanners with $\tilde{O}\left(f^3  n^{1+\frac{1}{k}} \right)$ many edges. Finally, one may use an $f$-EFT $k$-spanner which can be done in $O(f n^{1+(1/\lceil k/2\rceil)})$ for multigraphs \cite{ChechikLPR10,BodwinP19}, and $O(kf^{\frac{1}{2}}n^{1+(1/\lceil k/2\rceil)} + fn)$ for simple graphs. Moreover, each of the aforementioned results, improves upon the others for some range of $f$, which sums up into an unclear picture on the ``overhead'' one needs to pay for faults in $(\alpha,\beta)$-spanners. Meanwhile, for lower bounds, a folklore lower bound, a multigraph with $\Omega(fn^{1+\frac{1}{k}})$ edges and no proper $f$-EFT $(k,k-1)$-spanner, is known\footnote{The lower bound is conditional on the Erd\H{o}s girth conjecture. The conjecture states that for all integers $k$ there exist an $n$ node graph with $\Omega(n^{1+\frac{1}{k}})$ edges and girth $>2k+1$. It is common in the area of FT spanners to base lower bounds on the Erd\H{o}s girth conjecture, see e.g. \cite{BodwinDPW18,BodwinHP24,PetruschkaST24_ITCS} and reference therein.}.

Our first main result is clarifying the picture for $f$-EFT, $(k,k-1)$-spanners in \emph{multigraphs}, by giving a tight upper bound matching the best known lower bound, resolving a long-standing gap, and improving the $\Tilde{O}(f^3)$ dependency on $f$ by the Dinitz-Krauthgamer approach, to a tight $O(f)$ overhead.  We prove:

\begin{theorem}\label{thm:FT-spanner-bounds-intro}
Let $G$ be a multigraph and let $k,f\in\mathbb{N}$. Then:
\begin{enumerate}[leftmargin=*,itemsep=2pt]
  \item \textbf{Existence.} There exists an $f$-EFT $(k,k\!-\!1)$-spanner $H$ of $G$ with $O\!\big(f\,n^{1+1/k}\big)$ edges.
  \item \textbf{Construction.} There is a polynomial-time algorithm that outputs an $f$-EFT $(k,k\!-\!1)$-spanner $H$ of $G$ with $O\!\big(k f\,n^{1+1/k}\big)$ edges.
\end{enumerate}
\end{theorem}

We highlight that spanners are never
used with $k$ larger than $O(\log n)$, since the additional stretch does not bring additional sparsity.
On the other hand, the fault
parameter $f$ can be significantly larger, for example polynomial in $n$. In light of this, the improvement of \Cref{thm:FT-spanner-bounds-intro} on the current known bounds is substantial, and achieves stretch of an $(k,k-1)$-spanner under faults, with the same upper bound formerly known only for the weaker multiplicative $(2k-1)$ fault tolerant spanners, see \Cref{table:size-bounds}.

We emphasize two natural reasons to focus on multigraphs: (1) We believe that they are more natural than simple graphs in the presence of faults. This is reflected e.g. in the fact that when adding colors to the edges and allowing color faults, the bounds for FT-spanners in multigraphs do not change. On the other hand, when introducing color faults to a simple graph one gets exactly the same behavior as multigraphs, e.g. the upper/lower bounds for edge-color-fault tolerant spanners in \emph{simple} graphs jumps to $\Omega(fn^{1+\frac{1}{k}})$, matching the bound for edge faults/ edge-color faults in multigraphs\footnote{See  \cite{PetruschkaST24_ITCS} for a more detailed discussion on the best bounds on the various fault models, and the effect of introducing colors.}.  (2) Even for the well-studied \emph{multiplicative} EFT spanners, tight lower bounds are known for multigraphs but not for simple graphs, see \cite{BodwinDR22}.

\begin{table}
\centering
\scalebox{.8}{
\begin{tabular}{|c|c|c|c|c|}
\toprule
$(\alpha, \beta)$ & \textbf{Size bound}                                                                                                                  & \textbf{Ref.}\\ \midrule\midrule
    $(k+\varepsilon,k-1)$ & $O \left(f\cdot \left(\frac{(k-1)}{\varepsilon} \right)^{2f} n^{1+\frac{1}{k}} \right)$ & \cite{BraunschvigCPS15}\\
    \midrule
    $(k,k-1)$ & $\tilde{O}\left(f^3  n^{1+\frac{1}{k}} \right)$ &\cite{DinitzK11}\\
    \midrule
    $(k,0)$ & $O(f n^{1+(1/\lceil k/2\rceil)})$ & \cite{ChechikLPR10,BodwinP19}\\
    \midrule
    $(2k-1,0)$ & $O(f n^{1+1/k})$ & \cite{ChechikLPR10,BodwinP19}\\
    \midrule
    $(k, k-1)$ & $O( f n^{1+\frac{1}{k}})$ & \textbf{new}

      \\ \bottomrule
\end{tabular}
}
\caption{Size bounds for $f$-EFT $(\alpha, \beta)$-spanners on multigraphs.}\label{table:size-bounds}
\end{table}

Our second main result concerns weighted graphs. Spanners are often used to compress metric spaces that correspond to weighted
input graphs, see e.g. \cite{CaiK97,DobsonB14,MarbleB13,SalzmanSAH14}. Standard multiplicative spanners can handle weights easily, yet obtaining different better stretches for further pairs is not as obvious. The work of Cohen \cite{Cohen00} constructed weighted $(1+\varepsilon,\beta)$ spanners where the additive term corresponds to a $+ \beta \cdot w_{max}(G)$ approximation of the distance on top of the multiplicative stretch, where $w_{max}(G)$ denotes the maximum edge weight over all edges in $G$, which was then improved by \cite{Elkin01}. In recent years, there has been a renewed and growing interest in understanding spanners for weighted graphs. The work of \cite{ElkinGN22} achieved constructions of weighted $(\alpha,\beta)$ spanners with \emph{local} stretch, meaning that $w_{max}(G)$ is replaced by the maximum $w_{max}(P)$, which stands for the maximum edge weight along a shortest path $P$ between nodes. This was later further studied for \emph{additive} stretch in a sequence of works obtaining essentially tight results \cite{AhmedBSKS20,ElkinGN23,AhmedBHKS21,LaL24}.

While for unweighted graphs, an $(\alpha,\beta)$ spanner has  a better stretch than an $(\alpha',\beta')$ spanner with $\alpha<\alpha'$ for far enough vertex pairs, this is not necessarily the case in the weighted setting as the $+ \beta w_{max}$ term can significantly outweigh the weight of the path. In particular, while the line of works \cite{AhmedBSKS20,ElkinGN23,AhmedBHKS21,LaL24} gives tight bounds for additive spanners, such spanners may require $\Omega(n^{\frac{4}{3}-\varepsilon})$ edges for any $\epsilon>0$ by the seminal work of Abboud and Bodwin \cite{DBLP:journals/jacm/AbboudB17}. The state of weighted spanners with $O(n^{1+\frac{1}{k}})$ edges and local stretch guarantee is less understood, with essentially a single construction of Elkin Gitlitz and Neiman \cite{ElkinGN22} of a weighted $\left(O(1),k^{O(1)} \right)$- spanner with $O(n^{1+\frac{1}{k}})$ edges. Our second result aims to fill in two gaps that weighted $\left(O(1),k^{O(1)} \right)$- spanner with $O(n^{1+\frac{1}{k}})$ have: (1) the bound is ineffective for paths $P$ of \emph{hop length} $O(k^{O(1)})$, in the sense that a $(2k-1)$ multiplicative spanner produces a better stretch, and (2) In case of very irregular weights in $G$, the $+k^{O(1)}w_{max}(P)$ term may significantly outweigh the constant multiplicative stretch guaranteed. We complement their result by filling these gaps. We show an algorithm that takes a graph $G$ and returns a subgraph $H$ of size $O(n^{1+\frac{1}{k}})$ which \emph{strictly} improves over the stretch of the multiplicative greedy algorithm for all non-adjacent vertices, and obtains optimal approximation for vertices whose shortest path consists of two edges. Given a path $P$ of length $\ell$, let $w(P)$ denote the sums of weights along the path, i.e. $w(P)=\sum_{e\in P} w(e)$, and similarly define $w^{(1/2)}(P)$ the sum of the $\lceil \frac{\ell}{2}\rceil$ highest weights along the path $P$. We prove:

\begin{theorem}\label{thm:spanners-for-weighted-graphs-intro}
For every weighted graph $G$, there exists a subgraph $H$ with $O\!\big(n^{1+1/k}\big)$ edges such that for every pair $x,y$ joined by a path $P$ in $G$,
\[
  \dist_H(x,y) \;\le\; w(P) + (2k-2)\,w^{(1/2)}(P).
\]
Moreover, for shortest paths of hop-length $2$ this bound is best possible.
\end{theorem}

To our knowledge, such tight bounds in \Cref{thm:spanners-for-weighted-graphs-intro} were not known, even in the unweighted setting. In the unweighted case, the above result implies that pairs of vertices at distance $2$ in $G$ have stretch $k$ in $H$. The work of Baswana et al. \cite{BaswanaKMP10} achieves an upper bound of $O(kn^{1+\frac{1}{k}})$ for $(k,k-1)$ spanners while the work of Parter \cite{Parter14} achieved stretch $k$ for vertices at distance $2$ (in the unweighted case) and size $O(k^{2}n^{1+\frac{1}{k}})$, which was improved recently and independently by Chechik and Lifshitz \cite{ChechikL26} to a spanner with size $O(n^{1+1/k}+kn)$ which achieves stretch $k$ for pairs at distance $2$ and is optimal for $k=O(\log(n)/\log\log(n))$\footnote{Their construction also achieves stretch $k$ for pairs at distance $3$. We do not follow this path, but it is not hard to show that this can be achieved by the techniques presented in this paper, which implies the existence of hybrid spanners (see \cite{Parter14}) and optimal size $O(n^{1+1/k})$ assuming the girth conjecture. }. Nevertheless in our viewpoint handling weighted graphs is the more interesting part of \Cref{thm:spanners-for-weighted-graphs-intro}, whereas the $k$ factors is a secondary improvement.

Finally, we show that our technique is robust and useful for analyzing other variants of the greedy spanner. The parallel greedy spanner algorithm, introduced by Haeupler, Hershkovitz and Tan \cite{HaeuplerHT23}, is a variant of the standard greedy algorithm of \cite{AlthoferDDJS93} for computing a $(2k-1)$ spanner, where instead of considering an edge $e_i$ at the $i^{th}$ step, a set of edges $M_i$ which forms a matching is being considered, and each edge of $M_i$ decides for itself greedily if it could be discarded from $G$, or if it will be added to $H$ (the decision is done by checking for $e=\{u,v\}$ if $\dist_{H_i}(u,v)>2k-1$, where $H_i$ is the spanner before $M_i$ was considered). The parallel greedy spanner algorithm appears in relation to later works on length-constrained expanders and routing algorithms, see \cite{ChuzhoyP25,HaeuplerH025}. The main result of \cite{HaeuplerHT23} is that the parallel greedy algorithm adds at most $O(n^{1+O(\frac{1}{k})})$ edges overall, which was later improved to $O(k)^k \cdot O(n^{1+\frac{1}{k}})$ by Bodwin Haeupler and Parter \cite{BodwinHP24}. Both papers are technically heavy, with \cite{HaeuplerHT23} ``length-constrained'' expander decomposition machinery, and \cite{BodwinHP24} used the counting/dispersion approach   of \cite{BodwinDR22} together with a complex choice of paths to count over\footnote{The work of \cite{BodwinHP24} does not specifically target parallel greedy spanners but the more challenging bounded-degree FT-spanners, yet their results imply upper bounds on the output size of the parallel greedy spanner.}. We show that by applying the tools we develop in this work, one gets an elementary analysis (at most two pages) of the following improved bound:

\begin{proposition}\label{prop:parallel-greedy-spanner-intro}
    The output of the parallel greedy spanner algorithm with stretch parameter $(2k-1)$ has $O(kn^{1+\frac{1}{k}})$ many edges, and arboricity $O(kn^{1/k})$.
\end{proposition}

We highlight that this tighter analysis of the parallel greedy spanner, was independently obtained in \cite{BodwinHHT26}. They show that the improved bounds on the parallel greedy spanner imply length constrained expander decompositions with improved size bounds. 

\paragraph{Acknowledgments}
We are deeply grateful to Merav Parter for her invaluable guidance, encouraging this collaboration, and suggesting the problem of constructing better FT $(\alpha,\beta)$-spanners. We thank Asaf Petruschka for clarifying the connection between the bounded-degree fault model and the parallel greedy spanner, and we thank Nathan Wallheimer and Ron Safier for useful discussions.

\subsection{Technical Overview}

The starting point of this work, as well as a central component of the FT-spanner we construct, is a new greedy algorithm we introduce and analyze.

\subsubsection{The Greedy $d\to r$ Spanner Algorithm}
Our goal will be to study the following problem: given a graph $G$ and integers $d,r$, find a sparse subgraph $H$, such that vertices at distance exactly $d$ in $G$ are at distance $\leq r$ in $H$:

\begin{definition}[$d \to r$ spanner]
    Let $d,r$ be natural numbers. We say that a subgraph $H$ of $G$ is a $d \to r$ spanner of $G$ if $\forall x,y \in V$:  
    \[\dist_G(x,y) = d \implies \dist_H(x,y) \leq r\]
\end{definition}

We highlight that this definition is similar to the definition of $f(d)$-spanner appearing in many works and refer the reader to the survey \cite{AhmedBSHJKS20} and reference therein. Given $f:\mathbb{N}\to\mathbb{N}$ an $f(d)$ spanner $H$ of $G$ is a subgraph $H$ which is a $d\to f(d)$ spanner for all $d$, hence the main distinction in our approach is specializing only in fixed $d,f(d)$. $d \to r$ spanners can be thought of as building blocks for $(\alpha,\beta)$ spanners in particular we highlight the work of Parter \cite{Parter14} giving an almost tight construction of $2\to2k$ spanners, and the later work of Parter and Ben-Levy \cite{Ben-LevyP20} which constructs $k^{\varepsilon}\to O_{\varepsilon}(k)$ spanners with $O_{k,\varepsilon}(n^{1+\frac1k})$ edges for $0<\varepsilon<1$, see \cite{ChechikL26} for a recent improvement for a wide range of $\varepsilon$. We emphasize that the union of a few $d\to r$ spanners usually make an $(\alpha,\beta)$-spanner and in particular it is folklore (and easy to prove) that a union of a $1\to 2k-1$ and a $2 \to 2k$  spanner of a graph $G$, always produces a $(k,k-1)$ spanner. We will construct spanners by analyzing the following simple greedy $d \to r$ spanner algorithm:

\begin{algorithm}[h]
\caption{\textsc{Greedy $d$ $\to$ $r$ Spanner}$(G,d,r)$}\label{alg:greedy-hybrid-spanner}
$H \leftarrow (V,\emptyset)$\;

\For{$x,y\in V$ with $\dist_G(x,y) = d$}{
  \If{$\dist_H(x,y)>r$}{
    Add to $H$ an arbitrary $d$-path $P_{x,y}$ in $G$ connecting $x$ and $y$\;
  }
}
\Return{$H$}\;
\end{algorithm}

\paragraph{Large Cliques in the Output of the $d\to r$ Greedy Algorithm.} 

While the high girth of the greedy spanner of \cite{AlthoferDDJS93} is a main feature used in most works, the output of the $d\to r$ greedy algorithm for $d>1$ may have very small girth. To demonstrate this, we show that the output of the $2 \to 2k$ greedy algorithm does not necessarily have high girth, e.g. it can have a clique of size $\sqrt{n}$. Consequently, getting a meaningful size bound for the output of the algorithm requires more than just girth.

\begin{claim}
    There exists a graph $G$ on $n$ vertices and an order of the $2$-paths in it such that when the greedy $2\to 2k$ algorithm processes the paths in order, the output has a clique of size $\Omega(\sqrt{n})$.
\end{claim}

\begin{proof}
    Set $n=t^2$, and consider a clique of $t$ vertices $v_1, ..., v_t$. Connect each $v_i$ to $t-1$ other vertices $u_{i,1}, ..., u_{i,t-1}$, where all $u_{i,j}$ are distinct and so each a leaf in the resulting graph, see Figure~\ref{fig:big_clique}. Consider the following sequence of $2$-paths: $((v_{i+j}, v_i, u_{i,j}))_{i\in [t],j\in [t-1]}$ in the lexicographical order of $(i,j)$, with the indices of $v_i$ considered modulo $t$. 
    
    By construction, each new path introduces a fresh leaf, ensuring it will be kept by the greedy algorithm, forcing all clique edges into the output (see \Cref{fig:big_clique}).
    This shows that the greedy algorithm is forced to include every clique edge. The union of these paths contains the entire graph, and in particular, a $t$-clique.
\end{proof}

\begin{figure}[h]
    \centering    \includegraphics[width=0.5\linewidth]{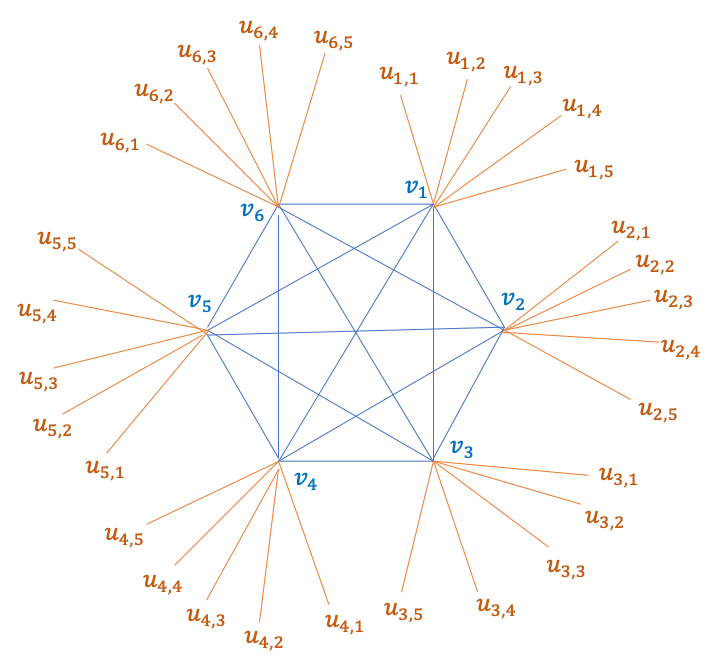}
    \caption{A big clique in the output of the $2\to 2k$ greedy spanner}
    \label{fig:big_clique}
\end{figure}

\paragraph{Greedy Clustering}
To bypass the lack of high girth we take a new approach, based on \emph{clustering} see e.g \cite{ElkinP04,BaswanaKMP10,BaswanaS07,Chechik13, Ben-LevyP20}. A central definition we use is that of a clustered vertex in a graph: a vertex $v$ is $\ell$-clustered in a graph $G$ with $n$ nodes if for all $i\leq \ell$, $|B_{G}(v,i)|\geq n^{i/k}$\footnote{Notice the definition of clusters depend on $n$ and $k$, though we omit $k$ from the definition for simplicity, as it is fixed throughout the paper.}. This definition is inspired by the new ``local'' viewpoint of the Baswana-Sen algorithm of \cite{ParterPST25}. We show that adding an edge between relatively far vertices in a graph helps them become more clustered, which is captured by the following key lemma:

\begin{lemma}[Neighborhood exchange lemma]\label{lem:ne-exchange-intro}
Let $(u,v)$ be vertices in a graph $H$, and assume that $\dist_H(u,v)>s$. Let $H'=H\cup \{u,v\}$. Then for every $\ell\le \lceil \frac{s}{2} \rceil-1$:
\begin{enumerate}[label=(\alph*),leftmargin=*,itemsep=2pt]
\item $\Ball_{H}(u,\ell)\cap \Ball_{H}(v,\ell+1)=\emptyset$.
\item $\Ball_{H'}(u,\ell)\subseteq \Ball_{H'}(v,\ell+1)$.
\end{enumerate}
\end{lemma}

The merit of the lemma above is that it implies good bounds on the following algorithm to cluster a graph in a greedy fashion: For an input graph $G$, initialize an empty graph $H$ and iteratively consider an edge $e \in G$, and add $e$ to $H$ if: (1) The endpoints of $e$  are at mutual distance $>s$ in $H$, and (2) one of $e$'s endpoints is not $\lceil \frac{s}{2} \rceil$ clustered in $H$. Using the lemma we show that this process adds only $O(sn^{1+\frac{1}{k}})$ edges where $k$ is the parameter of clustering. In particular if one inserts to this process the edges of a graph $G$, one obtains a graph $H$ as output, together with the following invariant which is often critical in clustering methods\footnote{See e.g. the local analysis of the Baswana-Sen algorithm in \cite{ParterPST25}.}: for every $\{u,v\}=e \in E(G)-E(H)$ either (1) $u,v$ are both $\lceil \frac{s}{2} \rceil$ clustered or (2) $\dist_H(u,v) \leq s$\footnote{One can also extend this process naturally to weighted graphs, which we use in the main result of \Cref{sec:wtd2to2k}}. 
We then apply this lemma to show the edges added by the \emph{parallel} greedy algorithm helps vertices cluster rapidly   -- we obtain essentially tight bounds for parallel greedy spanners. The size bound of $O(sn^{1+\frac{1}{k}})$ on the number of edges added by the process has a factor $s$ which is present because the analysis of clustering using \Cref{lem:ne-exchange-intro} goes in $s$ rounds. 
It turns out that sequential algorithms enjoy a better upper bound than the one achieved by this round-by-round clustering. In \Cref{sec:greedy-spanner-warm-up} we prove that when edges in the greedy clustering algorithm are added sequentially, the number of edges added is $O(n^{1+\frac{1}{k}})$.

\paragraph{Obtaining Tight Bounds for the $2\to2k$ Greedy Spanner:} 

Our approach is based on combining the clustering ideas formerly described with a ``path-buying''  inspired algorithm, describing how distances shrink throughout the algorithm, together with an extra step we call \emph{lateral clustering}. We consider the paths $P_1,\ldots,P_t$, with $P_i=(x_i,m_i,y_i)$ added by the $2\to 2k$ greedy algorithm, in their order. For each $P_i$ we consider the pair $(x_i,m_i)$ which was at distance $>k$ before adding $P_i$ to the graph (such a pair must exist by the triangle inequality), we assume without loss of generality this pair is $(x_i,m_i)$ (otherwise reverse the path) and call it the $i^{th}$ distant pair. Concretely if $x_i$ or $m_i$ is not $\lceil \frac{k}{2} \rceil$-clustered, we can ``charge'' adding $P_i$ as the addition of an edge to the greedy clustering algorithm with distance parameter $k$ -- as the sequential process adds $O(n^{1+\frac{1}{k}})$ edges, there are only $O(n^{1+\frac1k})$ paths with a distant pair with both endpoints not $\lceil \frac{k}{2} \rceil$-clustered.\footnote{In other words, we use a \emph{moderate} girth $>k+1$ subgraph for the analysis of the clustering that occur during the greedy $2\to 2k$ algorithm.} E.g. in \Cref{fig:big_clique} the orange edges are distant and contribute to the clustering of their endpoints, while the blue edges of the clique are typically not distant and do not contribute to clustering \emph{greedily}, nevertheless we can still use such paths for clustering in a more intentional way, we describe now. 

The rest of the paths have distant pairs which are \emph{both} already $\lceil \frac{k}{2} \rceil$- clustered, let $P_i$ be such path. If $y_i$ was also $\lceil \frac{k}{2} \rceil$-clustered we will be done by a ``path-buying''  like argument (Originally appearing in \cite{BaswanaKMP10}) -- we can show that the $\Omega(n^{\frac{k-1}{k}})$ pairs of vertices in $B_H(x_i,\lfloor \frac{k}{2} \rfloor) \times B_H(y_i,\lceil \frac{k}{2} \rceil) $ are at distance $>k$ before adding $P_i$ but $\leq k$ after adding $P_i$, as this can happens at most once for a pair of vertices, and as there are $O(n^2)$ pairs of vertices this will apply the desired upper bound for such paths. 
While on a high level, the approach above describes well the main ideas that go into the proof, unfortunately it does not work as stated -- there is no guarantee that $y_i$ is $\lceil \frac{k}{2} \rceil$ or even $1$-clustered! 
To handle this, the analysis requires an additional technical step which we call \emph{lateral} clustering. In the lateral clustering step we analyze the effect of adding $P_i$ on the size of the $\lceil \frac{k}{2} \rceil$-ball around $y_i$ and divide into the following cases: (1) Most vertices of the $(\lceil \frac{k}{2} \rceil-1)$-ball around $m_i$ were initially (before adding $P_i$) at distance $>\lceil \frac{k}{2} \rceil$ from $y_i$ -- in this case we still get a significant increase for the size of $B_H(y_i,\lceil \frac{k}{2} \rceil)$ by adding $P_i$, hence the number of steps for which such paths can increase the size of a $\lceil \frac{k}{2} \rceil$ around $y_i$ until it has a ball of size $\geq n^{\lceil k/2 \rceil}$ is $O(n^{1+1/k})$ getting us back to the path-buying type argument, as if $y_i$ is clustered. (2)  Most vertices of the $(\lceil \frac{k}{2} \rceil-1)$-ball around $m_i$ were initially (before adding $P_i$) at distance $\leq \lceil \frac{k}{2} \rceil$ from $y_i$.  In this case we apply a similar ``path-buying''  type argument for balls around $x_i$ and $m_i$, which concludes the proof.

In \Cref{sec:greedy-sqrt-spanner}, we show that the greedy algorithm can also be used for other values of $d$. We reprove a result of Ben-Levy and Parter (\cite{Ben-LevyP20}, Thm 1.2) of the existence of $\lceil \sqrt{k} \rceil\to O(k)$ spanners with $\tilde{O}(n^{1+\frac{1}{k}})$ edges, by showing the greedy algorithm with the same parameters achieves the same size bound.

\subsection{Handling Faults}

To handle faults, we analyze a  FT-greedy algorithm, originally due to \cite{BodwinDPW18}, adapted to the $d\to r$ spanner definition:

\begin{definition}[EFT $(d\to r)$ spanners]\label{def:eft}
Let $\G=(\V,\E)$ be a (multi)graph and $f\in\mathbb{N}$.
A subgraph $H\subseteq \G$ is an \emph{$f$-edge-fault tolerant} ($f$-EFT) \emph{$(d\to r)$ spanner} if for every $F\subseteq \E$ with $|F|\le f$ and all $x,y\in \V$:
    \[
        \dist_{G-F}(x,y)=d \implies \dist_{H-F}(x,y)\leq r
    \]
\end{definition}
We then study the FT-greedy spanner, which works as follows: For all $x,y$ go over all paths $P$ with $d$ edges between $x$ and $y$, and add $P$ to $H$ whenever there exist a fault set $F$ of $\leq f$ edges, disjoint from $P$, such that $\dist_{H-F}(x,y)>r$. We show that the algorithm above achieves the size bounds appearing in \Cref{thm:FT-spanner-bounds-intro}.

\paragraph{The Blocking Set Method for $d\to r$ Spanners}
In \cite{BodwinP19}, Bodwin and Patel studied VFT-spanners, and prove that the $f$-VFT $1\to2k-1$ greedy algorithm achieves optimal size bound. To do so, define a blocking set as a collection of edge-vertex pairs, such that each short ($\leq 2k$) cycle in $H$ contains one of the pairs. The FT greedy $1\to 2k-1$ algorithm naturally induces a blocking set of size $f|E(H)|$. They then show that a graph with such blocking set has $O(f^{1-\frac{1}{k}})$ edges for the vertex fault model. On an intuitive level, a graph that admits a small blocking set is often thought of as being close to having high girth, forcing it to be sparse.

In contrast we view the output of the greedy $d \to r$ as producing graphs such that many vertices have cuts of size $f$ that substantially separate them. A spark of this viewpoint appears in \cite{ParterT25}, where a tight analysis for a FT-greedy algorithm was given by applying a ``cut-based''  view instead of the cycle-based view of \cite{BodwinHP24}. By defining the $d \to r$ blocking set for the $d \to r$ FT-greedy algorithm, we show that the output size produced isn't too large when compared to the standard $d\to r$ greedy spanner. We show that if $b(n)$ is a bound on the maximal number of paths added by the greedy $2\to 2k$ spanner\footnote{We actually need a bound for a slightly different algorithm then \Cref{alg:greedy-hybrid-spanner}, nevertheless this is done for technical reason elaborated in \Cref{sec:ft-spanners}.} on an $n$-vertex graph. Then the $f$-FT $2\to2k$ greedy spanner adds at most $O(f\cdot b(n) + fn^{1+\frac{1}{k}})$ paths. This concludes the proof using the bounds obtained for the $2\to 2k$ spanner described previously.

The proof of the theorem above is short, and has two steps: (1) Showing that it is enough to prove the theorem when the FT algorithm adds edge-disjoint paths at each step, and (2) Showing that if the FT-greedy algorithm adds disjoint paths, then they may be thought of as ``long edges''  -- which reduces the dependence on the fault parameter to be linear similarly to \cite{BodwinP19, PetruschkaST24_ITCS}. 

In \Cref{sec:poly-time} we implement the modified FT-greedy approach of Dinitz and Robelle \cite{DinitzR20} to give a polynomial time construction as stated in \Cref{thm:FT-spanner-bounds-intro}.

\subsection{Improving Upon Multiplicative $2k-1$ Spanners for Weighted Graphs}

It is not hard to see that \Cref{thm:spanners-for-weighted-graphs-intro} would follow if we could show that any weighted graph $G$ has a subgraph $H$ with $O(n^{1+\frac{1}{k}})$ edges and the following property: any $2$-path $P$ in $G$ has a replacement path $P'$ in $H$ of weight $w(P)+(2k-2)w_{max}(P)$ which is what we achieve in \Cref{sec:wtd2to2k}. We also highlight that the section begins with a simple lower bound, \Cref{clm:lower-bound-wtd}, showing that the stretch $w(P)+(2k-2)w_{max}(P)$ is the best possible stretch that can be achieved by a subgraph with $O(n^{1+\frac{1}{k}})$ edges, assuming the girth conjecture.
 
Following the theme presented in the paper, one may guess the following generalization of the greedy $2\to 2k$ algorithm in the spirit of \cite{AlthoferDDJS93} - go over the paths $P$ in order of weight and add each path which does not have good approximation in the current subgraph and add every such path. The main difficulty in this approach stems from the following question: According to which order should $P$ be processed?  The multiplicative greedy spanner scans edges, and in such case there is a natural order - order of weight. On the other hand the greedy $d\to r$ spanner for $d>1$ scans paths of length $d$ -- if those were composed of weighted edges it is unclear which order to choose: maybe order the paths by maximum edge weight on the path? another natural options are the total edge weight/ a linear combination of the weights along the path. We didn't find a single order for which the analysis goes through.

A different attempt is to use what is called $d$-light initialization, appearing in \cite{AhmedBSKS20,AhmedBHKS21,ElkinGN23} in the context of weighted additive spanners. $d$-light initialization describes an initialization phase prior to running an algorithm, in which every vertex adds its $d$ lightest adjacent edges - in our case $O_k(n^{\frac{1}{k}})$ edges. This approach is effective in handling $k\leq 4$ but doesn't seem to generalize to higher $k$ - the main reason is that it allows to compare weights in a very local way (around a vertex), which does not propagate to something meaningful for long enough paths. Even if one replaces the initialization phase by a clustering phase, it seems that additional structure is needed to design the algorithm.

In light of the above, we took a different approach - instead of running a greedy algorithm in some unique order, such that an analysis similar to \Cref{thm:main} follows, we break the algorithm according to the steps in \emph{the analysis} of \Cref{thm:main}, by introducing weighted notions of clustering, lateral clustering, a greedy phase, and a new "distance-reduction" phase, each with its own novel order. We believe that this approach can be fruitful for other constructions of spanners for weighted graphs.

We now briefly describe the phases of the algorithm, their corresponding order, and what they informally try to achieve. The spanner will be the union of all edges added in the various phases. Assume for now $k$ is even, and $R=k/2$. The other parity of $k$ is handled similarly in \Cref{sec:wtd2to2k}.

\paragraph{Phase 1: Adding a Multiplicative Spanner} We begin by adding to the spanner a $(2k-1)$ multiplicative spanner with $O(n^{1+1/k})$ edges. The main reason for this phase is the following observation: every $2$-path $P$ (in $G$) with an edge in the final spanner (the second edge may not be present in the final spanner) has the right stretch in the final spanner. Hence it's enough to obtain a good stretch for $2$-paths in $G$ with both edges not in the final spanner. 

\paragraph{Phase 2: Greedy Clustering}
We go over all edges $e$ by increasing order of weight and will run the greedy clustering algorithm on all edges by order of weight and add those added by the process, which is $O(n^{1+1/k})$ edges. By the end of this phase we ensure that any edge $e=\{u,v\}$ not in the final spanner has either (1) $\dist_H(u,v)\leq kw(e)$, or (2) that $u,v$ are clustered at the time of considering $e$ -- which implies that $B_{H_{\leq w(e)}}(u, i)$, and $B_{H_{\leq w(e)}}(v, i)$ are of size $\geq n^{\frac{i}{k}}$, for $i \leq R$ \footnote{$H_{\leq w(e)}$ is the underlying of all edges of weight at most $w(e)$.}. In case an edge wasn't added due to condition (2) we call this edge saturated.

It now follows that paths $P$ with two edges that weren't included because $\dist_H(u,v)\leq kw(e)$ have good stretch already by the end of this phase. For the rest of the paths $P=(x,m,y)$ which do not meet the desired stretch by the end of phase $2$ we have (w.l.o.g) that both $x,m$ were clustered at the time (aka weight) for which the edge $\{x,m\}$ was considered in phase $2$. This allows to add edges that improve distances between vertices, while still adding only $O(n^{1+1/k})$ edges.

\paragraph{Phase 3: Lateral Clustering}
The lateral clustering phase is an algorithm where each vertex $y$ (independently) adds $O(n^{1/k})$ edges. For every vertex $v$ denote by $w_v$ to be the smallest weight of a saturated edge (see description of phase 2) incident to $v$. Each vertex $y$ goes over its neighbors $u$ by order of weight $\phi = w(y,u)+(R-1) w_u$ and checks whether adding the edge $\{y,u\}$ substantially increases the size of the ball of radius $\phi$ around $y$. The reason for choosing this particular order, is that this way we order $y$'s neighbor in a monotone way according to the implied distance (to $y$) guarantee we can get on the vertices from $u$'s ball, after the edge addition.

After phase 3 the neighbors of $y$ that were not added divide into two: (type 1 neighbors) neighbors $u$ of $y$ such that at time of considering $u$ the $\phi$-ball around $y$ was already of size $\Omega(n^{R/k})= \Omega(n^{1/2})$, and (type 2 neighbors) neighbors $u$ such that when considering $u$, the ball of radius $(R-1)w_u$ in the spanner, shares many neighbors with the ball of radius $\phi$ around $y$.

\paragraph{Phase 4: Global Distance Reduction by Edges}
In this phase we go over all \emph{saturated} edge $e=\{u,v\}$, that is edges that weren't added by phase $2$ due to both endpoints being clustered (according to the corresponding weight) and reconsider adding them to the spanner - as they may create useful shortcuts for vertex pairs between their clusters. More concretely, we go over $e=(u,v)$ by order of weight and for each $e$ we check how many pairs of $B(u,R w_u) \times B(v,(R-1)w_v)$ will be of distance $\leq kw(e)$ due to adding $e$ - if this quantity is $\Omega(n^{1-1/k})$ we will add $e$ to the spanner. Again this phase adds the right number of edges. The key property that holds by the end of this phase is that we have successfully handled all paths $P=(x,m,y)$ such that: (1) $(x,m)$ was saturated in phase 2, and wasn't added during phase 4 (2) when $y$ considered $m$ in phase 3, $m$ was a type 2 neighbor when considered by $y$.

\paragraph{Phase 5: Greedy Phase}
The fifth and last phase is greedy, we will go over all paths $P$ that do not have the right stretch in the spanner, and if a path doesn't satisfy the right stretch we add it. It is not hard to check that the only path $P=(x,m,y)$ that may with a bad stretch satisfies: (1) $(x,m)$ is saturated and wasn't added in phase 4, and (2) $m$ was a type 2 neighbor of $y$ (meaning that $y$ already has a sufficiently large ball around it). In such case, we show that many distances ($\Omega(n^{1-1/k})$ in the set $B_H(x,(R-1)w(x,m)) \times B_H(y,w(y,m)+(R-1)w_m)$ must improve by adding the path. Still we need to make sure that every pair  of vertices improve at most once. For this we choose the last order to be $2w(m,y)+(k-1)w(x,m)$ - which correspond to the distance bound of the short paths created by the greedy fixes.

\subsection{Organization of the Paper}

In \Cref{sec:greedy-spanner-warm-up} we set the stage and develop the key definitions and tools we need through the paper, and warm up by showing how to apply these tools to analyze the parallel greedy spanner. In \Cref{sec:analysis} we analyze the $2\to2k$ greedy algorithm. In \Cref{sec:wtd2to2k} we describe and analyze the weighted spanner in \Cref{thm:spanners-for-weighted-graphs-intro}. \Cref{sec:ft-spanners} concludes the construction of the $f$-FT $(k,k-1)$-spanner with size $O(fn^{1+\frac{1}{k}})$. 

\subsection{Notations}\label{subsec:notations}
For an unweighted graph $G$ we denote for $x,y\in V(G)$ the length of their shortest $x$ to $y$ path $P$ by $\dist_G(x,y)$. For a path $P=(v_1,\ldots v_{\ell+1})$ we use $x(P)$ to denote the first vertex of the path $v_1$ and similarly $y(P)$ denotes the last vertex on the path, and $\ell$ is the length of the path. We denote by $\Ball_G(v,\ell)$ the set of vertices $u\in V(G)$ that has a path of length $\leq \ell$ to $v$. 

A weighted graph $G$ is a triple $(V,E,w)$ with $w: E \to \mathbb{R}_{>0}$. For a path $P$ in a weighted graph we denote by $w(P)$ the sum of weights of $E(P)$, by $w_{\max}(P)$ (resp.\ $w_{\min}(P)$) the maximum (resp.\ minimum) edge weight along $P$. For a weighted graph $G$,  $\dist_G(x,y)$ denotes the minimum of $w(P)$ over all $x$ to $y$ paths in $G$. For a weighted graph $G$ we write $G^\circ$ for the underlying unweighted graph of $G$.
For a threshold $\omega>0$, let $\Gleq{\omega}$ be the subgraph of $G$ induced by edges of weight at most $\omega$, and let $\GleqCirc{\omega}$ be its underlying unweighted graph.  For two sets $A,B$ we denote $A-B$ and $A/B$ the corresponding difference set. For a set $F\subseteq E$ and graph $G=(V,E)$ denotes $G-F$ the graph $(V,E-F)$.

\section{Clustering in Changing Graphs}\label{sec:greedy-spanner-warm-up}

We now introduce tools that would be useful throughout the paper, and in particular, describe the greedy clustering algorithm, and prove Proposition~\ref{prop:parallel}.

We begin with a definition of a clustered vertex. The definition above depends on $k$, and throughout the paper all clusters will appear with the same $k$ which will be clear from context.

\begin{definition}[Clusters]\label{def:clusters_greedy}
A vertex $v$ has an $\ell$-cluster in a graph $H$ if for all $r\le \ell$ we have
\(
\lvert \Ball_{H}(v,r)\rvert \ge n^{r/k}.
\)
If $\ell$ is maximal with this property, we say $v$ is \emph{$\ell$-clustered}.
\end{definition}

Intuitively, being $\ell$-clustered means that the neighborhood of a vertex grows at least polynomially fast in radius, with growth $n^{\frac{1}{k}}$ per step.

Notice that any vertex is $0$-clustered in every graph. We also have the following observation.

\begin{observation}\label{cl:no-k-clustered}
    If a vertex $v$ is $k$-clustered in a graph $G$, then $ \forall y \in V: \dist_G(v,y)\leq k$.
\end{observation}

\begin{proof}
    By definition a $k$-clustered vertex $v$ has $|B_G(v,k)|\geq n^{\frac{k}{k}}=n$. Hence $B_G(v,k)=V$ and the claim follows. 
\end{proof}

We next restate and prove the neighborhood exchange lemma. This lemma is crucial as it shows how clustering properties propagate when adding a new edge between distant vertices.

\begin{lemma}[Neighborhood exchange lemma]\label{lem:ne-exchange}
Let $(u,v)$ be vertices in a graph $H$, and assume that $\dist_H(u,v)>s$. Let $H'=H\cup \{u,v\}$. Then for every $\ell\le \lceil \frac{s}{2} \rceil-1$:
\begin{enumerate}[label=(\alph*),leftmargin=*,itemsep=2pt]
\item $\Ball_{H}(u,\ell)\cap \Ball_{H}(v,\ell+1)=\emptyset$ and  $\Ball_{H}(u,\ell+1)\cap \Ball_{H}(v,\ell)=\emptyset$ .
\item $\Ball_{H'}(u,\ell)\subseteq \Ball_{H'}(v,\ell+1)$, and $\Ball_{H'}(v,\ell)\subseteq \Ball_{H'}(u,\ell+1)$.
\end{enumerate}
\end{lemma}

\begin{proof}
    For (a), suppose $w\in \Ball_{H}(u,\ell)\cap \Ball_{H}(v,\ell+1)$. Then
\(
\dist_{H}(u,w)\le \ell
\)
and
\(
\dist_{H}(w,v)\le \ell+1
\),
so by the triangle inequality,
\(
\dist_{H}(u,v)\le 2\ell+1 \le s,
\)
contradicting that $\dist_{H}(u,v) > s$. For (b), since $(u,v)\in H'$, for every $w\in B_{H'}(u,\ell)$ we have
\[
\dist_{H'}(v,w) \le \dist_{H'}(v,u)+\dist_{H'}(u,w) \le 1+\ell,
\]
and this implies the claim.
\end{proof}

\subsection{Greedy Clustering}
For a fixed integer $s$ (that is different in different contexts) and a subgraph $H$ of $G$ we shall call a vertex which has an $\lceil{\frac{s}{2}}\rceil$-cluster in $H$ \emph{fully clustered} in $H$ and call a pair of vertices $(u, v)$ \emph{distant} if $\dist_H(u,v) > s$. When $s$ is clear from context we simply say a vertex is fully-clustered, and that a pair of vertices is distant, according to the above.

\noindent Consider the following algorithm:
\begin{algorithm}[h]
\caption{\textsc{Greedy Clustering}$(G, s, (e_i)_{i=1}^t)$}\label{alg:boost_proc}

$H \leftarrow (V(G), \emptyset)$\;

\For{$i = 1,..., t$}{
    \If{$e_i=\{u_i,v_i\}$ satisfies: (1) $\dist_H(u_i,v_i)>s$, and (2) $u_i$ or $v_i$ is not $
    \lceil \frac{s}{2} \rceil$ clustered in $H$:}{
        $E(H) \leftarrow E(H)\cup \{e_i\}$
    }
}
\Return{$H$}\;
\end{algorithm}

We call an edge $e = \{u,v\}$ boosting if $(u,v)$ is a distant pair and either $u$ or $v$ is not fully clustered at the time of consideration. \Cref{alg:boost_proc} appears throughout this paper both explicitly as a step of an algorithm (also for weighted graphs, as it appears in \Cref{sec:wtd2to2k} ) and implicitly as a subsequence of actions performed by the greedy algorithm that can be coupled with execution of \Cref{alg:boost_proc}. We now state and prove an upper bound on the number of edges that can be actually added by \Cref{alg:boost_proc}.

\begin{proposition}
\label{lem:boosting_upper_bound}
    The number of edges added by \Cref{alg:boost_proc} is $O(n^{1+\frac1k})$ for any $s$.
\end{proposition}

For the proof we will need the following:

\begin{lemma}\label{lem:key-lem-boost-girth}
     Let $G=(V,E)$ be a graph with $n$ vertices, and $E=\{e_1,\ldots,e_t\}$ with $e_i=\{u_i,v_i\}$. Let $G_i$ be the graph $(V,\{e_1,\ldots,e_{i-1}\})$.
     Assume moreover the following holds:

     \begin{enumerate}
         \item $G$ has girth $>s+1$.

         \item For every $i$, either $u_i$ or $v_i$ does not have $\lceil \frac{s}{2} \rceil $-cluster in $G_i$.
     \end{enumerate}
     Then $t=O(n^{1+\frac1k})$.
 \end{lemma}

\begin{proof}
    Let $D$ be the average degree of $G$. Recall that we can always find an induced subgraph $G'=(V',E')$ of $G$ of minimum degree $D'$ with $\frac{D}{2}\leq D'$, by iteratively removing low degree vertices. We claim $D' < n^{\frac{1}{k}}+2$. Assume towards contradiction that $D'\geq  n^{\frac{1}{k}}+2$. Let $e$ be the last edge from $E(G)$ appearing in $G'$, and  $G''=(V',E'-e)$, and observe that clearly we have that the minimum degree of $G''$ is at least $n^{\frac{1}{k}}+1$. Since $G$ has girth $>s+1$ the same holds for the subgraph $G''$. Moreover since $G'$ has girth $>s+1$, the BFS tree around every vertex has no collisions up to radius $\lceil \frac{s}{2} \rceil $. We obtain that for any vertex $v\in V', \ell\leq \lceil \frac{s}{2} \rceil $ we have \[|B_{G''}(v,\ell)| \geq 1+ \sum_{i=0}^{\ell-1} D'' \cdot (D''-1)^{i} \geq 1+n^{\frac1k}\sum_{i=1}^{\ell-1} n^{i/k} \geq n^{\ell/k}.\]
    In particular it means that any vertex in $G''$ is $\lceil \frac{s}{2} \rceil$-clustered including the endpoints of $e$ -- contradiction to assumption 2 of the lemma.
    Hence we conclude that $\frac{D}{2}\leq D' \leq n^{\frac{1}{k}}+2$. As $D=\frac{2|E|}{|V|}$ we obtain $|E|= O(n^{1+\frac{1}{k}})$.

\end{proof}

\begin{proof}[Proof of \Cref{lem:boosting_upper_bound}.]
    Consider the graph on vertex set $V(G)$ induced by all the edges added by \Cref{alg:boost_proc}. We claim that this graph has girth $>s+1$. 

    Assume the contrary: there is a cycle of length $\leq s+1$. Consider the edge $e_i = \{u_i, v_i\}$ that was the last edge in this cycle added by the algorithm. Then all other edges of the cycle were already present in $H$ at the moment when the algorithm considered $e_i$ and they formed a path between $u_i$ and $v_i$, and so $\dist_H(u_i, v_i) \leq s$, which contradicts condition (1) of adding $e_i$.
    
    Now we see that the sequence of edges added by \Cref{alg:boost_proc} satisfies the conditions of \Cref{lem:key-lem-boost-girth}: the high girth condition is just shown and the non-clustered condition is implied by condition (2) of adding an edge. So the desired bound follows. 
\end{proof}

Now we will state and prove a result for the parallel greedy spanner. The parallel greedy spanner is defined algorithmically as follows.

\begin{algorithm}[h]
\caption{\textsc{Greedy parallel $1$ $\to$ $2k-1$ Spanner}$(G,k, \{\mathcal{M}_i\}_{i=1}^t)$}\label{alg:parallel}

$H \leftarrow (V,\emptyset)$\;

\For{$i = 1,..., t$}{
    $S \leftarrow \emptyset$
    
    \For{$e = uv \in \mathcal{M}_i$}{
        \If{$\dist_H(u,v)>2k-1$}{
        $S \leftarrow S\cup \{e\}$
        }
    }
    
    $E(H) \leftarrow E(H)\cup S$
}
\Return{$H$}\;

\end{algorithm}

\begin{definition}[Parallel greedy $1\to 2k-1$ spanner]
    In a graph $G$ on $n$ vertices fix some (ordered) sequence of matchings $\mathcal{M}_1, ..., \mathcal{M}_t$. We call the output $H$ of the Algorithm~\ref{alg:parallel} on input $(G,k, \{\mathcal{M}_i\}_{i=1}^t)$ a parallel greedy $1\to 2k-1$ spanner of $G$.
\end{definition}

We shall prove the following:

\begin{proposition}\label{prop:parallel}
    Any parallel greedy $1\to 2k-1$ spanner of an $n$-vertex graph $G$ has at most $O(kn^{1+\frac{1}{k}})$ edges.
\end{proposition}
\begin{proof}
    We first note that by \Cref{cl:no-k-clustered},  \Cref{alg:parallel} never adds an edge that contains a $k$-clustered vertex (we assume $k>1$ otherwise the process is trivial). This shows that the parallel spanner is the same as the result of the following procedure.
    Take \Cref{alg:boost_proc} with $s = 2k-1$ and modify it as follows: instead of a sequence of edges use a sequence of matchings in $G$ and on each step add all the edges from the current matching that are boosting. Since now we add many edges simultaneously, the girth argument from the previous proof doesn't work and we will use the neighborhood exchange lemma instead. 

    Suppose that on the step $i$ we add an edge $uv$, where $u$ is $\ell_u$-clustered and $v$ is $\ell_v$-clustered and $\ell_u \leq \ell_v$, in such case we say that a boost occur at the vertex $u$. By definition of being $\ell_u$-clustered, $\lvert B_{H_i}(u,\ell_u + 1)\rvert < n^{\frac{\ell_u + 1}{k}}$, where $H_i$ is the current $H$ before the $i$-th step. On the other hand, recall that by the \Cref{cl:no-k-clustered}, $\ell_u \leq k-1$, and so by the Lemma~\ref{lem:ne-exchange} with $s = 2k-1$, $B_{H_i}(v,\ell_u) \subseteq B_{H_{i+1}}(u,\ell_u+1) \backslash B_{H_{i}}(u,\ell_u+1)$, and so $\lvert B_{H_{i+1}}(u,\ell_u+1)\rvert - \lvert B_{H_{i}}(u,\ell_u)\rvert \geq \lvert B_{H_i}(v,\ell_u)\rvert \geq n^{\frac{\ell_u}{k}}$, where the last inequality follows from $v$ being $\ell_v \geq \ell_u$-clustered.
    
    Such a thing can clearly occur to each vertex $u$ for fixed $\ell_u < k$ on at most $n^{\frac{1}{k}}$ many iterations of \Cref{alg:parallel}  until $u$ gets $(\ell_u + 1)$-clustered, so there are  $O(kn^{\frac{1}{k}})$ iterations for which $u$ is boosted in total (a vertex can't be $>k$ clustered). Since the edges added on each iteration are vertex disjoint, the number of boosts on each step is at least the number of added edges, and so the number of added edges is bounded by $n\cdot O(kn^{\frac{1}{k}})$ as required.
\end{proof}

\begin{remark}
    We highlight that the proof above also shows that the \emph{arboricity}\footnote{The arboricity of a graph $G$ is the minimum integer $a$, such that $G$ is a union of $a$ forests. It is well known to be equivalent (up to a constant factor), to the minimum $b$ such that the edges of $G$ can be oriented s.t. every vertex has $\leq b$ ingoing edges.} of the output graph obtained by the parallel greedy spanner is $O(kn^{1/k})$. This is evident by orienting the edge $e=\{u,v\}$ from the non-boosted vertex to the boosted one (with ties broken arbitrarily), and the aforementioned bound on the number of steps in which a vertex $v$ can be boosted throughout the process.  
\end{remark}

\paragraph{A Limitation for Removing the $k$ Factor Above:} We note that the upper bound above can't get sharper once $k>\log(n)$, hence one may also assume $k=O(\log(n))$. We note that in general the $k$ factor can't be removed (in contrast to the standard greedy algorithm). Consider $n=2^k$, and the $n$ vertex graph of the $k$-hypercube $Q_k$ and $\mathcal{M}_i$ to be the edges parallel to $e_i$. Then the endpoints of each edge in the new matching are disconnected in the current spanner, and so the algorithm adds all $k\cdot 2^k= \Omega(kn^{1+\frac{1}{k}})$ edges.

\paragraph{Spanner for Path Collections}
For the construction of the FT spanners we construct we actually need to analyze a slightly more general greedy $d\to r$ spanner for \emph{path collections}. This is mainly for a technical reasons related to handling faults elaborated in \Cref{rmk:ft-path-collection}.
Let $\mathcal{P}$ be a collection of paths of length $d$ on the vertex set $V$. For a sub-collection $\pp'$ denote by $G'=(V,\cup_{P' \in \pp'} E(P'))$. A subcollection $\pp'$ is $d\to r$-spanner for $\pp$ if $\forall P \in \pp$:  $\dist_{G'}(x(P),y(P))\leq r$. The size of the spanner is just $|\pp'|$. 

\noindent We now describe the corresponding greedy construction of spanners for path collections.

\begin{algorithm}[h]
\caption{\textsc{Greedy $d$ $\to$ $r$ Spanner for path collections}$(V,\mathcal{P},r)$}\label{alg:greedy-spanner-path-collections}
$H \leftarrow (V,\emptyset)$\;

\For{$P\in \mathcal{P}$ with endpoint $x,y$}{
  \If{$\dist_H(x,y)>r$}{
    Add $E(P)$ to $H$\;
  }
}
\Return{$H$}\;
\end{algorithm}

\begin{observation}\label{obs:path-collections-imply-graph-spanners}
    If an algorithm returns for any path collection $\pp$ on $n$ vertices a $d\to r$ spanner of size $\leq b(n)$, then there is an algorithm that given a graph $G$ outputs a $d\to r$ spanner with $\leq d b(n)$ many edges.
\end{observation}

\begin{proof}
    Given a graph $G$ apply the let $\pp_d(G)$ be all length $d$ paths in $G$. Let $\pp'$ be a $d\to r$ spanner of $\pp_d(G)$, then $(V,\cup_{P'\in \pp'} E(P') )$ is a $d \to r$ spanner of $G$ of the required size.
\end{proof}
In light of this, it is enough to obtain bounds for spanner for path collections. The main distinction is that paths in $\pp$ do not have to be shortest in any form in the definition above.

\section{Analysis of the Greedy $2\to 2k$ algorithm}
\label{sec:analysis}

The main goal of this section is to prove the following:

\begin{theorem}\label{thm:main}
The greedy $2 \to 2k$ algorithm for path collections, \Cref{alg:greedy-spanner-path-collections}, Outputs a $2 \to 2k$ spanner with $O\!\left(n^{1+\frac1k}\right)$  $2$-paths, and runs in polynomial time.
\end{theorem}

By \Cref{obs:path-collections-imply-graph-spanners} this would also prove that any graph with $n$ vertices has a $2\to2k$ spanner of size $O(n^{1+\frac{1}{k}})$. 

On a high level, the proof is based on showing that whenever a $2$-path is added, distances in the current $H$ shrink in a structured way. Initially a path addition affects local geometry, but over time each addition produces a larger global effect.

Let $\PiPath=(x_i,m_i,y_i)$ denote the $i$-th $2$-path added by the algorithm, and let
\[
\Hi \;=\; \bigl(V,\;\bigcup_{j<i} P_j\bigr)
\]
be the subgraph consisting of all edges that were added \emph{before} adding the current path $\PiPath$. 
\begin{claim}\label{clm:split-k}
If $\PiPath=(x_i,m_i,y_i)$ is added, then either $\dist_{\Hi}(x_i,m_i)>k$ or $\dist_{\Hi}(m_i,y_i)>k$.
\end{claim}

\begin{proof}
Since $\PiPath$ was added, $\dist_{\Hi}(x_i,y_i)>2k$. By the triangle inequality,
\[
\dist_{\Hi}(x_i,m_i)+\dist_{\Hi}(m_i,y_i)\;\ge\;\dist_{\Hi}(x_i,y_i)\;>\;2k,
\]
so at least one of the two summands exceeds $k$.
\end{proof}

We call $(x_i,m_i)$ an \emph{$i$-distant pair} if $\dist_{\Hi}(x_i,m_i)>k$. Thus every added path yields at least one distant pair. We assume wlog $(x_i,m_i)$ is the distant pair, otherwise -- reverse the order of the path.

\paragraph{Local Growth Radius.}
Set
\[
R \;=\;
\begin{cases}
\frac{k+1}{2}, & \text{if $k$ is odd},\\[2pt]
\frac{k}{2},   & \text{if $k$ is even}.
\end{cases}
\]
Think of $R$ as the radius up to which we bound the \emph{local} effect of path additions. When $R$-balls become sufficiently large, we switch to a different, global counting argument. Set also $I_{odd}$ to be $1$ if $k$ is odd and $0$ otherwise. This technical quantities are used to unify the notation for two parities. Notice that $2R - I_{\mathrm{odd}} = k$ for all $k$.

\paragraph{Clusters.}
We use the definition of clusters from Section~\ref{sec:greedy-spanner-warm-up}, see Definition~\ref{def:clusters_greedy}, and we need the following addition to it to simplify the terminology.

\begin{definition}[Fully clustered vertices]\label{def:clusters}
If $v$ has an $R$-cluster, we say $v$ is \emph{fully clustered}.
\end{definition}

This coincides with the definition given in Section~\ref{sec:greedy-spanner-warm-up} with $s = k$.

\paragraph{Boosting Steps.}
An $i$-distant pair $(u,v)$ is \emph{boosting} if at least one of $\{u,v\}$ is \emph{not} fully clustered in $\Hi$. We call the step $i$ of the algorithm boosting if at least one $i$-distant pair is boosting. This coincides with the definition of boosting from Section ~\ref{sec:greedy-spanner-warm-up} with $s = k$.

\begin{lemma}\label{cor:total-boosting}
The total number of boosting steps over all vertices is $O\!\left(n^{1+\frac1k}\right)$.
\end{lemma}

\begin{proof}
    It follows from the definitions and the fact that having a cluster in the subgraph implies having a cluster in the ambient graph that if we invoke the \Cref{alg:boost_proc} with $s = k$ on the list of all boosting edges in order of addition it will add all of them. Then \Cref{lem:boosting_upper_bound} gives the desired bound.
\end{proof}

\begin{remark}\label{rem:mono}
Cluster sizes are (weakly) monotone: since $\Hi\subseteq H_{i+1}$, we have $\Ball_{\Hi}(v,r)\subseteq \Ball_{H_{i+1}}(v,r)$ for all $v,i,r$. Thus once $v$ has an $\ell$-cluster, it keeps it thereafter.
\end{remark}

\subsection{Lateral Clustering}
\label{subsec:lateral}

The discussion above accounts for boosting steps, which imply that the distant edge of each remaining path has clustered endpoints. In particular, for the argument to go through we would like to show that even the vertex on $P_i$ which \emph{doesn't participate} in the distant pair will have a big ball around it. To handle this we introduce:

\begin{itemize}[leftmargin=*,itemsep=2pt]
\item \emph{Lateral clusters}: formed around the \emph{not-necessarily-distant} endpoint of the newly added $2$-path.
\item A \emph{global} counting argument once clusters are full which resembles the path-buying technique.
\end{itemize}

\begin{definition}[Lateral clusters]\label{def:lateral}
A vertex $v$ has an $\ell$ \emph{lateral cluster} in $\Hi$ if $\lvert \Ball_{\Hi}(v,\ell)\rvert \ge n^{\ell/k}$.
\end{definition}

From now on, when $\PiPath=(x_i,m_i,y_i)$ is not boosting and $(x_i,m_i)$ is the $i$-distant pair, we examine lateral clusters around $y_i$. 

\begin{definition}[Lateral boosting]\label{def:lateral-boost}
We say a \emph{lateral boost} occurs at step $i$ (for the non-distant endpoint $y_i$) if:
\begin{enumerate}[label=(\roman*),leftmargin=*,itemsep=2pt]
\item $y_i$ is not $R$-laterally clustered in $\Hi$; and
\item $\bigl|\{\,u\in \Ball_{\Hi}(m_i,R-1):\ \dist_{\Hi}(y_i,u)>R\,\}\bigr|\ \ge \ \tfrac12\,n^{(R-1)/k}$.
\end{enumerate}
In this case, we say $y_i$ is \emph{laterally boosted}.
\end{definition}

\begin{claim}\label{clm:lateral-growth}
If $y_i$ is laterally boosted at step $i$, then
\[
\bigl|\Ball_{H_{i+1}}(y_i,R)\setminus \Ball_{\Hi}(y_i,R)\bigr|\ \ge\ \tfrac12\,n^{(R-1)/k}.
\]
\end{claim}

\begin{proof}
All vertices in $\{u\in \Ball_{\Hi}(m_i,R-1):\ \dist_{\Hi}(y_i,u)>R\}$ lie outside $\Ball_{\Hi}(y_i,R)$ but lie inside $\Ball_{H_{i+1}}(y_i,R)$ after adding $P_i$, by concatenating $y_i\!\to\!m_i$ and a shortest path in $H_i$ from $m_i$ to $u\in \Ball_{\Hi}(m_i,R-1)$.
\end{proof}

Therefore we have:

\begin{claim}\label{clm:lateral-bound-per-vertex}
Every vertex $y$ can be laterally boosted at most $2n^{1/k}$ times.
\end{claim}

\begin{proof}
    If $y_i$ is laterally boosted then
\[
\big|\Ball_{H_{i+1}}(y_i,R)\setminus \Ball_{H_i}(y_i,R)\big|
\;\ge\; \tfrac12\,n^{(R-1)/k},
\]
hence any vertex has an $R$-ball of size $n^{R/k}$ after $2n^{1/k}$ lateral boosts.
\end{proof}

\begin{corollary}[Non-boosting steps that are laterally boosted]\label{cor:lateral-count}
The number of indices $i$ for which the $i$-distant pair is not boosting \emph{and} a lateral boost occurs is $O\!\left(n^{1+1/k}\right)$.
\end{corollary}

\bigskip

We now bound all \emph{remaining} non-boosting steps (steps that do not have a boosting distant pair, and do not trigger a lateral boost).
There are two different reasons a lateral boost may fail; we handle each via a global counting argument that tracks many pairs whose distances cross a fixed threshold in one step.

\medskip
Define for an integer $t\ge 0$ the distance-threshold set
\[
S_t \;=\; \{(u,v,i)\in V\times V\times \mathbb{N} \;:\; \dist_{H_i}(u,v)>t\ \text{ and }\ \dist_{H_{i+1}}(u,v)\le t\}.
\]
Monotonicity of distances implies each ordered pair $(u,v)$ can contribute to $S_t$ for at most one index $i$, hence $|S_t|=O(n^2)$.

\begin{lemma}\label{lem:no-lateral-because-1}
Let $P_i=(x_i,m_i,y_i)$ be a step whose $i$-distant pair is not boosting, and assume \emph{no} lateral boost occurs because \emph{condition (i)} fails (i.e., $y_i$ already has an $R$-lateral cluster in $H_i$). Then the number of such steps is $O\!\left(n^{1+1/k}\right)$.
\end{lemma}

\begin{figure}[h!]
    \centering
    \includegraphics[width=0.5\linewidth]{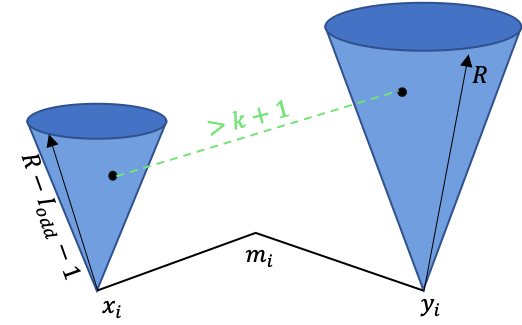}
    \caption{A path with an $R$-lateral cluster.}
    \label{fig:gr_global}
\end{figure}

\begin{proof}

As hinted in the above discussion, the general idea of the proof is to estimate the number of pairs of vertices in the lateral cluster of $y_i$ and some usual cluster of $x_i$ the distance between which shrinks by adding the path. See Figure~\ref{fig:gr_global}.

Fix such an index $i$ and consider the set
\[
\mathcal{Q}_i \;=\; \Ball_{H_i}(x_i,\,R-1-I_{\mathrm{odd}})\ \times\ \Ball_{H_i}(y_i,\,R).
\]

Because the step is non-boosting, $x_i$ is fully clustered, hence
\[
\big|\Ball_{H_i}(x_i,\,R-1-I_{\mathrm{odd}})\big| \;\ge\; n^{(R-1-I_{\mathrm{odd}})/k},
\qquad
\big|\Ball_{H_i}(y_i,\,R)\big| \;\ge\; n^{R/k}.
\]
Now, since $P_i$ was added, $\dist_{H_i}(x_i,y_i)>2k$. For any $v_x\in \Ball_{H_i}(x_i,R-1-I_{\mathrm{odd}})$ and $v_y\in \Ball_{H_i}(y_i,R)$, the triangle inequality gives
\[
\dist_{H_i}(v_x,v_y)
\;\ge\; \dist_{H_i}(x_i,y_i) \;-\; \dist_{H_i}(x_i,v_x) \;-\; \dist_{H_i}(v_y,y_i)
\;>\; 2k - \big((R-1-I_{\mathrm{odd}}) + R\big)
\;=\; k+1,
\]
using $2R-1-I_{\mathrm{odd}}=k-1$.
After adding $P_i$, there is a walk $v_x \to x_i \to y_i \to v_y$ of length at most $(R-1-I_{\mathrm{odd}})+2+R=k+1$ in $H_{i+1}$, so $\dist_{H_{i+1}}(v_x,v_y)\le k+1$.
Hence $\mathcal{Q}_i\times\{i\}\subseteq S_{k+1}$, and by the clustering assumptions in the statement
\[
|\mathcal{Q}_i|
\;\ge\; n^{(R-1-I_{\mathrm{odd}})/k}\cdot n^{R/k}
\;=\; n^{(k-1)/k}.
\]
Summing over all such $i$ and using $|S_{k+1}|=O(n^2)$ yields $O(n^{1+1/k})$ steps. 
\end{proof}

\begin{lemma}\label{lem:no-lateral-because-2}
Let $P_i=(x_i,m_i,y_i)$ be a step whose $i$-distant pair is not boosting, and assume \emph{no} lateral boost occurs because \emph{condition (ii)} fails, i.e.,
\[
\big|\{u\in \Ball_{H_i}(m_i,R-1): \dist_{H_i}(y_i,u)>R\}\big| \;\le\; \tfrac12\,n^{(R-1)/k}.
\]
Then the number of such steps is $O\!\left(n^{1+1/k}\right)$.
\end{lemma}

\begin{figure}[h!]
    \centering
    \includegraphics[width=0.5\linewidth]{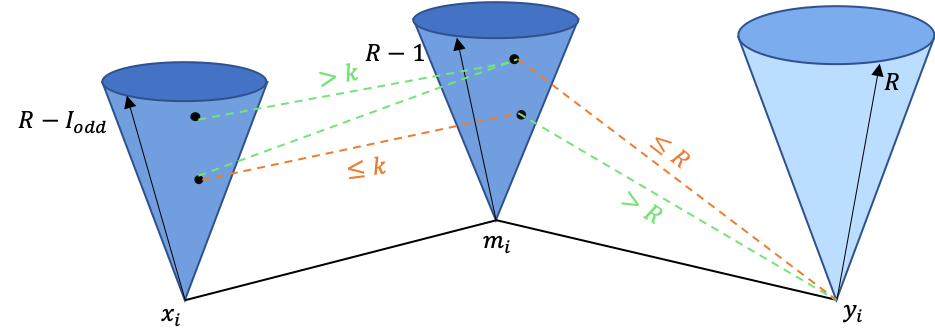}
    \caption{A path with a cluster of $m_i$ almost contained in the lateral cluster of $y_i$. Observe that for any added path, two orange edges as depicted in the figure cannot share a vertex.}
    \label{fig:gr_lateral}
\end{figure}

\begin{proof}

The idea is again to consider many pairs of vertices with shrinking distances, this time in the clusters of $x_i$ and $m_i$, see Figure~\ref{fig:gr_lateral}.

Let $C := \{u\in \Ball_{H_i}(m_i,R-1): \dist_{H_i}(y_i,u)\le R\}$.
Since $m_i$ is fully clustered and the “far” set in the statement of the claim has size at most $\tfrac12 n^{(R-1)/k}$, we have
$|C|\ge \tfrac12\,n^{(R-1)/k}$.
Consider
\[
\mathcal{R}_i \;=\; \Ball_{H_i}(x_i,\,R-I_{\mathrm{odd}})\ \times\ C.
\]
Arguing as before, from $\dist_{H_i}(x_i,y_i)>2k$ and the triangle inequality, every $(v_x,v_m)\in \mathcal{R}_i$ satisfies
\[
\dist_{H_i}(v_x,v_m)
\;>\; 2k - \big((R-I_{\mathrm{odd}}) + R\big)
\;=\; k.
\]
In $H_{i+1}$ there is a walk $v_x\to x_i\to m_i\to v_m$ of length at most $(R-I_{\mathrm{odd}})+1+(R-1)=k$, hence $\dist_{H_{i+1}}(v_x,v_m)\le k$.
Thus $\mathcal{R}_i\times\{i\}\subseteq S_k$, and
\[
|\mathcal{R}_i| \;\ge\; n^{(R-I_{\mathrm{odd}})/k}\cdot \tfrac12\,n^{(R-1)/k}
\;=\; \tfrac12\,n^{(2R-1-I_{\mathrm{odd}})/k}
\;=\; \tfrac12\,n^{(k-1)/k}.
\]
Summing over $i$ and using $|S_k|=O(n^2)$ again yields $O(n^{1+1/k})$ steps.
\end{proof}


\begin{proof}[Proof of \Cref{thm:main}]
By the clustering analysis, the number of boosting steps is $O\!\left(n^{1+1/k}\right)$.
Among non-boosting steps, the laterally boosted ones contribute $O\!\left(n^{1+1/k}\right)$ by Corollary~\ref{cor:lateral-count}, while the remaining non-boosting, non-lateral steps contribute $O\!\left(n^{1+1/k}\right)$ by Lemmas~\ref{lem:no-lateral-because-1} and \ref{lem:no-lateral-because-2}.
Summing these bounds proves the claim.
The running time is polynomial since all 2-paths can be enumerated in $O(n^3)$ time and each distance check is linear in the current graph size.
\end{proof}

\begin{theorem}[Equivalent combinatorial formulation]\label{thm:combinatorial}
Let $G$ be an unweighted graph and let $P_1,\dots,P_t$ be 2-paths with $P_i=(x_i,m_i,y_i)$ such that
$\dist_{\bigcup_{j<i} P_j}(x_i,y_i)>2k$ for all $i$.
Then $t=O\!\left(n^{1+1/k}\right)$.
\end{theorem}

\begin{proof}
Processing the paths $P_1,\dots,P_t$ in order, each $P_i$ is added by the greedy rule by assumption, so the number of additions equals $t$.
Apply Theorem~\ref{thm:main}.
\end{proof}

\section{Improved Stretch for Weighted Graphs}\label{sec:wtd2to2k}


We refer the reader to \Cref{subsec:notations} to recall notations used for weighted graphs. We remind the reader the objective of this section.

\medskip
\noindent\textbf{Goal.}
Construct, in polynomial time, a weighted subgraph $H\subseteq G$ of size $O\!\left(n^{1+1/k}\right)$ such that for every $2$-edge path $P=(x,m,y)$ in $G$ there is an $x\!-\!y$ path in $H$ of weight at most
\[
  w(P) + (2k-2)\,w_{\max}(P).
\]
 
We note that the stretch bound above is optimal assuming the Erd\H{o}s girth conjecture as we now elaborate.

\begin{claim}\label{clm:lower-bound-wtd}
    Assuming the Erd\H{o}s girth conjecture holds, for each $n$ there is a weighted graph $G$ on $n>1$  vertices with $\Omega(n^{1+\frac{1}{k-1}})$ edges such that for all proper subgraphs $H$ there is a $2$-path $P = (x,m,y)$ in $G$ such that $\dist_H(x,y) \geq w(P) + (2k-2)w_{\max}(P)$. 
\end{claim}

\begin{proof}
    Take a graph $G'$ on $\frac{n}{2}$ vertices with $c'(\frac{n}{2})^{1+\frac{1}{k-1}} > \frac{c'}{4}n^{1+\frac{1}{k-1}}$ edges and of girth $>2(k-1) + 1$, promised by the Erd\H{o}s girth conjecture. We can assume that $G'$ is connected (otherwise connect components by a tree). Let all weights of the edges of $G'$ be $1$. Then for each vertex $x \in G'$ add a new vertex $x'$ and an edge $xx'$ of weight $0 < \varepsilon < 1$. Let $G$ be the resulting graph. Let $H$ be a proper subgraph of $G$. Suppose that for some $x\in G'$ we have $xx' \notin E(H)$. Then take a neighbor $y$ of $x$ in $G'$ (such exists because $G'$ is connected.) Then $\dist_{G}(x',y) = 2$, but $\dist_{H}(x', y) = \infty$, since $x'$ is a leaf in $G$. Thus we can assume that all leaf edges $xx'$ are in $E(H)$, and since $H$ is proper, there is some $e = xy \in  E(G')\backslash E(H)$. Take $P = (x',x,y)$, then $\dist_H(x',y) \ge \varepsilon + \dist_{G'-e}(x,y) \geq \varepsilon + (2k-1) = w(P) + (2k-2)w_{\max}(P)$, where the second inequality follows from the high girth assumption.
\end{proof}

\begin{remark}
    The example above shows that the stretch above is tight in the following sense: let $s$ any stretch function for $2$ paths $s:\mathbb{R}_{>0}\times \mathbb{R}_{>0}$ (each entry correspond to the edge weight of the first/ second edge along a path). If one can achieve size bound $o(n^{1+\frac{1}{k-1}-\varepsilon})$ for a spanner with stretch guarantee $s$, and assuming the girth conjecture, it must hold that $s(x,y) \geq x+y +(2k-2) \max(x,y)$, hence our goal is to achieve the optimal stretch for paths of hop-distance $2$.
\end{remark}

We keep the same notation for $R,I_{odd}$ as appears in \Cref{sec:analysis}. Also for the unweighted graphs that appear in this section we use the same terminology related to distances, balls and being (laterally) clustered as in the previous sections.

\subsection*{Algorithm Overview}

We describe a five-phase construction of $H$ which is inspired by the \emph{analysis} of the $2\to2k$ spanner appearing in \Cref{sec:analysis}. The algorithm has three main differences from the proof of \Cref{thm:main}: (1) We initially add a sparse $(2k-1)$-spanner to the graph $H$ we aim to construct. This step is routine, and is done in order to simplify the analysis. We remark that actually this step can be removed. (2) We add a new phase (Phase 4) between the clustering and lateral clustering parts of the analysis of \Cref{thm:main}, in which edges $\{u,v\}$ with distant \emph{and} fully-clustered endpoints may still be added in order to makes their clusters mutually close. This ensures that whenever $\{u,v\}$ is not added in this phase, replacement paths can traverse from the cluster of $u$ to the cluster of $v$ without explicitly going to $u,v$ reducing their length. (3) We sort edges/paths by three different carefully chosen keys in three different phases. One key is used for clustering, another is used for lateral clustering, and the last key is used in the last phase (Phase 5) where we greedily add paths that do not have the stretch we want.

We now proceed with describing the algorithm in detail, and analyze along each phase the number of added edges, as well as certain invariants that are maintained for edges/paths that were not added by the phase. 

\paragraph{Phase 1: Initialization.}
Initialize $H$ to contain any $(2k\!-\!1)$-spanner of $G$ with $O(n^{1+1/k})$ edges (e.g., by the standard greedy construction).
This guarantees that any $2$-path having at least one edge already in $H$ attains the target stretch.

\begin{observation}\label{obs:edge-in-H}
For any $2$-path $P=(x,m,y)$, if either $\{x,m\}\in H$ or $\{m,y\}\in H$, then $\dist_H(x,y) \leq w(P)+(2k-2)w_{\max}(P)$.
\end{observation}

\paragraph{Phase 2: Clustering.}
Process edges $e=\{u,v\}\in E$ in nondecreasing order of weight $\omega=w(e)$.
Let $H_\le^\circ := \HleqCirc{\omega}$ be the underlying unweighted graph of the current $H$ restricted to edges of weight at most $\omega$.
If
\begin{enumerate}[label=(\roman*),nosep]
\item $\dist_{H_\le^\circ}(u,v) > k$, and
\item at least one of $u,v$ is \emph{not} fully clustered in $H_\le^\circ$,
\end{enumerate}
add $e$ to $H$.
An edge that \emph{met} (i) but was \emph{not} added (because both endpoints were already fully clustered) is called \emph{saturated}.
We call edges added in this phase  \emph{clustering edges}.

\begin{lemma}[Phase-2 size]\label{lem:phase2}
The number of clustering edges is $O\!\left(n^{1+1/k}\right)$.
\end{lemma}

\begin{proof}
    It follows from the definitions that the procedure \Cref{alg:boost_proc} with $s = k$ applied on the edges added by Phase 1 in order of addition puts all of them in the resulting graph. Then the desired bound follows from \Cref{lem:boosting_upper_bound}.
\end{proof}

\begin{claim}[Saturated edges]\label{clm:saturated}
If an edge $e=\{u,v\}$ is saturated at threshold $\omega=w(e)$, then both $u$ and $v$ are fully clustered already in the (unweighted) graph $\HleqCirc{\omega}$ formed by edges added before $e$ reached weight~$\omega$.
Consequently, if $w_u$ (resp.\ $w_v$) denotes the minimum threshold at which $u$ (resp.\ $v$) becomes fully clustered, then $w_u\le \omega$ and $w_v\le \omega$.
\end{claim}

\paragraph{Phase 3: Lateral Clustering.}
For each $v\in V$, process its neighbors $u\in N_G(v)$ with $u$ fully clustered, in increasing order of the key
\[
  \phi(v,u) \;=\; (R-1)\,w_u + w(v,u),
\]
where $w_u$ is the threshold at which $u$ first became fully clustered.
If $\card{\Ball_H\!\bigl(v,\phi(v,u)\bigr)} \ge n^{R/k}$, skip $u$. In this case we say the candidate $u$ is \emph{saturated} when considered by $v$;
Otherwise let
\begin{equation*}
     T \;:=\; \Ball_H\!\bigl(u,(R-1)w_u\bigr)\setminus \Ball_H\!\bigl(v,\phi(v,u)\bigr). 
\end{equation*}
\noindent and \emph{test} if $\card{T}>0.1\,n^{(R-1)/k}$; if so, add the edge $\{v,u\}$ to $H$.
Such an added edge is called a \emph{lateral clustering edge}.
If the test fails because $\card{T}\le 0.1\,n^{(R-1)/k}$, we say the cluster of $u$ is \emph{roughly contained} in the lateral cluster of $v$. See Figure~\ref{fig:phase3}.

\begin{figure}[h!]
    \centering
    \includegraphics[width=0.5\linewidth]{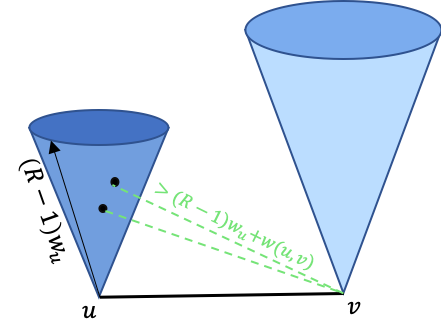}
    \caption{Lateral clustering  for weighted graphs. The edge $\{u,v\}$ was added in the ``$T$-test'' .}
    \label{fig:phase3}
\end{figure}

\begin{claim}[Phase-3 size]\label{clm:phase3}
Each vertex adds $O(n^{1/k})$ lateral clustering edges; hence Phase~3 contributes $O(n^{1+1/k})$ edges overall.
\end{claim}

\begin{proof}
Fix $v$ and order additions as $e_i=\{v,u_i\}$ with keys $\phi(v,u_1)\le \phi(v,u_2)\le\cdots$.
Let $H_i$ be the subgraph before adding $e_i$ and abbreviate
\(
  B_i := \Ball_{H_i}\!\bigl(v,\phi(v,u_i)\bigr).
\)
By triangle inequality and the processing order,
\[
  \Ball_{H_i}\!\bigl(u_i,(R-1)w_{u_i}\bigr) \subseteq \Ball_{H_{i+1}}\!\bigl(v,\phi(v,u_i)\bigr) \subseteq \Ball_{H_{i+1}}\!\bigl(v,\phi(v,u_{i+1})\bigr),
\]
hence $T \subseteq B_{i+1}\setminus B_i$ and $B_i\subseteq B_{i+1}$.
Whenever $e_i$ is added we have $\card{B_i}<n^{R/k}$ and $\card{B_{i+1}\setminus B_i}>0.1\,n^{(R-1)/k}$, so each $v$ contributes  at most $O(n^{1/k})$ edge additions in this phase. 
\end{proof}

\paragraph{Phase 4: Global Distance Reduction by Edges.}
Process edges $e=(u,v)$ of $G$ in nondecreasing order of weight $\omega=w(e)$.
If $\dist_{\HleqCirc{\omega}}(u,v)>k$, define
\[
  P \;=\; \Bigl\{(a,b)\in \Ball_{\HleqCirc{\omega}}\!\bigl(v,R-I_{\mathrm{odd}}\bigr)\times \Ball_{\HleqCirc{\omega}}\!\bigl(u,R-1\bigr)\;:\;
  \dist_{\HleqCirc{\omega}}(a,b) > k \Bigr\}.
\]
If $\card{P}>0.1\,n^{(k-1)/k}$, add $e$ to $H$ (and test symmetrically with $u,v$ swapped).
Edges added here are \emph{distance-reduction edges}.
An edge that met $\dist_{\HleqCirc{\omega}}(u,v)>k$ but was \emph{not} added has \emph{roughly close clusters}. See Figure~\ref{fig:phase4}.

\begin{figure}[h!]
    \centering
    \includegraphics[width=0.5\linewidth]{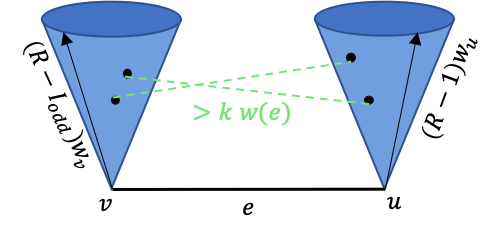}
    \caption{Global path reduction by edges.}
    \label{fig:phase4}
\end{figure}

\begin{claim}[Phase-4 size]\label{clm:phase4}
Phase~4 adds $O(n^{1+1/k})$ edges.
\end{claim}

\begin{proof}
For each addition index $i$ with edge $e_i=(u_i,v_i)$ at threshold $\omega_i=w(e_i)$, let $H_i$ be the pre-addition graph and consider the \emph{distance-threshold set}
\[
  S_k \;=\; \bigl\{(x,y,i):\; \dist_{\HleqCirc{\omega_i}}(x,y) > k \text{ and } \dist_{\HleqCirc{\omega_{i+1}}}(x,y) \le k \bigr\}.
\]
Since $\omega_i$ is non decreasing, each ordered pair $(x,y)$ contributes to $S_k$ for at most one $i$, hence $\card{S_k}=O(n^2)$.
If $P_i$ is the witness set above for $e_i$, then $P_i\times\{i\}\subseteq S_k$ and $\card{P_i}>0.1\,n^{(k-1)/k}$, so the number of additions is $O(n^{1+1/k})$. 
\end{proof}

\paragraph{A Key Combinatorial Lemma.}
We now capture a type of $2$-paths that achieve the desired stretch of the spanner due to the application of phases 3 and 4. Such a $2$-path, say $(x,m,y)$ has one edge $(x,m)$ with fully clustered ends which wasn't added in phase 4, and one edge $(y,m)$ which wasn't added in phase $3$ as $m$'s cluster was already roughly contained in $y$'s lateral cluster. We will show that in such a case the path $P$ has the desired stretch by finding  a vertex $u$ as depicted in Figure~\ref{fig:no-need-paths}.

\begin{figure}
    \centering
    \includegraphics[width=0.5\linewidth]{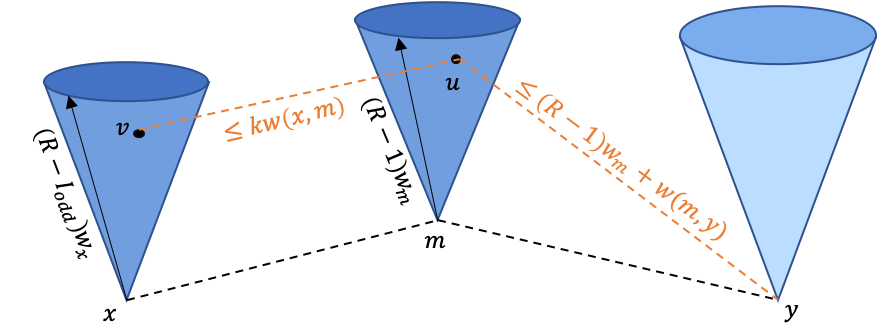}
    \caption{Paths that do not need to be added, and the corresponding replacement path.}
    \label{fig:no-need-paths}
\end{figure}

\begin{lemma}\label{lem:key-lemma}
Let $P=(x,m,y)$ be a $2$-path in $G$.
Assume that:
\begin{enumerate}[label=(\alph*),nosep]
\item $\{x,m\}$ has roughly close clusters (so both $x$ and $m$ are fully clustered and Phase~4 did not add $\{x,m\}$), and
\item the cluster of $m$ is roughly contained in the lateral cluster of $y$ (so Phase~3, when considered from $y$, did not add $\{y,m\}$ for failing the $T$-test).
\end{enumerate}
Let $H'$ be the subgraph consisting of all edges added prior to the (latest) consideration of either $\{x,m\}$ in Phase~4 or $\{y,m\}$ in Phase~3.
Then
\[
  \dist_{H'}(x,y) \;\le\; w(P) + (2k-2)\,w_{\max}(P).
\]
\end{lemma}

\begin{proof}
By (a), $x$ and $m$ are fully clustered.
Let $w_m$ be the threshold at which $m$ became fully clustered and set
\(
  B_m := \{u : \dist_{\HleqCirc{w_m}}(m,u) \le R-1\}.
\)
Then $\card{B_m}\ge n^{(R-1)/k}$.
Let $H''$ be the subgraph when $\{y,m\}$ was considered in Phase~3.
Define
\[
  F_y := \bigl\{u\in B_m \;:\; \dist_{H''}(y,u) > (R-1)\,w_m + w(y,m)\bigr\}.
\]
Since the $T$-test failed from $y$, $\card{F_y} \le 0.1\,n^{(R-1)/k}$.
Let $w_x$ be the clustering threshold of $x$, let $H'$ be as in the statement, and define
\[
  F_x := \Bigl\{u\in B_m \;:\; \forall v\in \Ball_{\HleqCirc{w_x}}(x,R-I_{\mathrm{odd}}),\;
  \dist_{\HleqCirc{w(x,m)}}(u,v) > k \Bigr\}.
\]
If $\card{F_x}$ were $>0.1\,n^{(R-1)/k}$ then, because $\card{\Ball_{\HleqCirc{w_x}}(x,R-I_{\mathrm{odd}})}\ge n^{(R-I_{\mathrm{odd}})/k}$ (full clustering of $x$),
the Cartesian product would give $>0.1\,n^{(k-1)/k}$ pairs at distance $>k$, forcing Phase~4 to add $\{x,m\}$—contradiction.
Hence $\card{F_x}\le 0.1\,n^{(R-1)/k}$.

Therefore, $\card{B_m\setminus(F_x\cup F_y)} \ge 0.8\,n^{(R-1)/k} \ge 1$, so we can pick $u\in B_m\setminus(F_x\cup F_y)$
and some $v\in \Ball_{\HleqCirc{w_x}}(x,R-I_{\mathrm{odd}})$ with $\dist_{\HleqCirc{w(x,m)}}(u,v)\le k$.
By construction, 
\[
  \dist_{H'}(y,u) \le (R-1)\,w_m + w(y,m),\qquad \dist_{H'}(x,v) \le (R-I_{\mathrm{odd}})\,w_x,\qquad \dist_{H'}(v,u) \le k\,w(x,m).
\]
Concatenating, see again \Cref{fig:phase4}, we get
\begin{align*}
 \dist_{H'}(x,y)
 &\le \dist_{H'}(x,v)+\dist_{H'}(v,u)+\dist_{H'}(u,y)\\
 &\le (R-I_{\mathrm{odd}})\,w_x + k\,w(x,m) + (R-1)\,w_m + w(y,m)\\
 &\le (2k-1)\,w(x,m) + w(y,m) \;\le\; w(P) + (2k-2)\,w_{\max}(P),
\end{align*}
where the third inequality uses $R-I_{\mathrm{odd}}+R-1 = k-1$ and $w_x\le w(x,m)$, $w_m\le w(x,m)$.
\end{proof}

\paragraph{Phase 5: Greedy Path Additions (Final Repairs).}
Let $\mathcal{P}$ be the set of $2$-paths $P=(x,m,y)$ that still lack a replacement path in $H$ of weight at most $w(P)+(2k-2)w_{\max}(P)$.

\begin{lemma}\label{lem:exists-sat}
Every $P\in\mathcal{P}$ contains a saturated edge.
\end{lemma}

\begin{proof}
If $E(P)\cap E(H)\ne\emptyset$, Observation~\ref{obs:edge-in-H} yields the required stretch, so for $P\in\mathcal{P}$ neither edge lies in $H$.
When each edge of $P$ was inspected in Phase~2, either it was saturated or $\dist_{\HleqCirc{w(e)}}(\cdot,\cdot)\le k$ held.
If the latter held for \emph{both} edges, then concatenating the two $k$-bounded subpaths gives
\(
  \dist_H(x,y)\le k\,w(x,m)+k\,w(m,y)\le w(P)+(2k-2)\,w_{\max}(P),
\)
 contradiction.
Thus at least one edge is saturated.
\end{proof}

By the above lemma, for each $P\in \mathcal{P}$ we can choose:
\begin{itemize}[nosep]
\item one saturated edge $e_{P,\mathrm{sat}}\in E(P)$, and
\item the other edge $e_{P,\mathrm{lat}}$ (``lat'' for \emph{lateral}).
\end{itemize}
Process the paths of $\mathcal{P}$ in increasing order of the key: 
\[2\,w(e_{P,\mathrm{lat}})+(k-1)\,w(e_{P,\mathrm{sat}})\]
When a path $P$ is encountered, if the desired replacement path is still absent, insert \emph{both} edges of $P$ into $H$.

\begin{claim}[Phase-5 size]\label{clm:phase5}
Phase~5 adds $O(n^{1+1/k})$ \emph{paths} (hence $O(n^{1+1/k})$ edges).
\end{claim}

\begin{figure}
    \centering
    \includegraphics[width=0.5\linewidth]{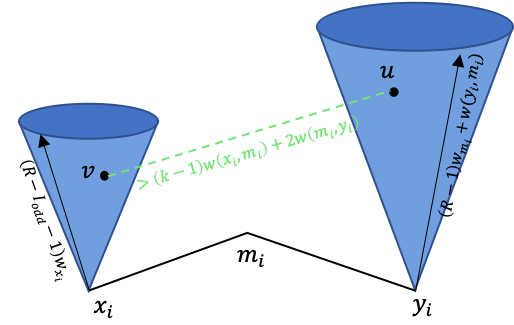}
    \caption{Greedy path additions}
    \label{fig:phase5}
\end{figure}

\begin{proof}
Let $P_i=(x_i,m_i,y_i)$ be the $i$-th added path and $H_i$ the current graph before adding $P_i$. Assume w.l.o.g that the saturated edge in this path is $\{x_i,m_i\}$.
By Lemma~\ref{lem:key-lemma}, for $P_i$ to be still ``bad'', when $\{y_i,m_i\}$ was considered from $y_i$ in Phase~3 it must have been \emph{saturated} (meaning that the ball of radius $w(m_i,y_i)+(R-1)w_{m_i}$ has many vertices).
Also, since $\{x_i,m_i\}$ is saturated, each of its endpoints has an $(R-1-I_{\mathrm{odd}})$-cluster.
Consequently,
\[
  \card{\Ball_{H_i}\!\bigl(y_i,(R-1)w_{m_i}+w(y_i,m_i)\bigr)} \;\ge\; n^{R/k}
  \quad\text{and}\quad
  \card{\Ball_{H_i}\!\bigl(x_i,(R-1-I_{\mathrm{odd}})w_{x_i}\bigr)} \;\ge\; n^{(R-1-I_{\mathrm{odd}})/k}.
\]
Consider the set 
\[
  S'_i \;=\;\Ball_{H_i}\!\bigl(x_i,(R-1-I_{\mathrm{odd}})w_{x_i}\bigr) \times  \Ball_{H_i}\!\bigl(y_i,(R-1)w_{m_i}+w(y_i,m_i)\bigr) .
\]
Define the distance threshold family
\[
S' \;=\; \left\{(u,v,i):\;
   \begin{aligned}
     &\dist_{H_i}(u,v) > (k-1)\, w(x_i,m_i)+2\,w(y_i,m_i), \\
     &\dist_{H_{i+1}}(u,v) \le (k-1)\,w(x_i,m_i)+2\,w(y_i,m_i)
   \end{aligned}
\right\}.
\]

It is hence enough to prove the following claim, that says that the naming $S_i'$ is justified -- that is that $S_i'$ is contained in the $i^{th}$ projection of $S'$. Together with the processing order of the keys, this implies that different indices $i$ yield disjoint $(u,v)$-projections, whence $\card{S'}=O(n^2)$, concluding the proof.

Observe, that for every $(u,v)\in S'_i$ the addition of $P_i$ creates a walk of length at most \[(R-1)w_{m_i}+w(y_i,m_i) + w(P_i) + (R-1-I_{\mathrm{odd}})w_{x_i} \le (k-1)\,w(x_i,m_i)+2\,w(y_i,m_i).\] On the other hand, if $\dist_{H_i}(u,v) \leq (k-1)\,w(x_i,m_i)+2\,w(y_i,m_i)$ then, by the triangle 
inequality:
\begin{align}
\dist_{H_i}(x_i, y_i) 
   &\leq \dist_{H_i}(x_i, u) + \dist_{H_i}(u, v) + \dist_{H_i}(v, y_i) \notag\\
   &\leq (R-1-I_{\mathrm{odd}})w_{x_i} 
        + (k-1)\,w(x_i,m_i) + 2\,w(y_i,m_i) 
        + (R-1)w_{m_i} + w(y_i,m_i) \notag\\
   &\leq (2k-3)w(x_i, m_i) + 3w(m_i,y_i) \notag\\
   &\leq w(P) + (2k-2)w_{\max}(P)
\end{align}
and so $P$ would not be added by the algorithm. This shows that $\dist_{H_i}(u,v) > (k-1)\,w(x_i,m_i)+2\,w(y_i,m_i)$ for all $(u,v) \in S_i'$.

So $S'_i\times\{i\}\subseteq S'$.
Since $\card{S'_i}\ge n^{R/k}\cdot n^{(R-1-I_{\mathrm{odd}})/k}=n^{(k-1)/k}$, the number of added paths is $O(n^{1+1/k})$. See Figure~\ref{fig:phase5}.
\end{proof}

\subsection*{Correctness and Size}

Recall that for a path $P$, $w^{(1/2)}(P)$ denotes the sum of the $\lceil \frac{\ell}{2}\rceil$ highest weights along $P$.

\begin{theorem}\label{thm:weighted-main}
Given a graph $G$ the five-phase construction outputs a subgraph $H$ such that:
\begin{enumerate}
    \item For any two vertices $x,y \in V(G)$ and a path $P$ between them we have $\dist_H(x,y) \leq w(P) + (2k-2) w^{(1/2)}(P)$;
    \item $|E(H)| = O\!\left(n^{1+1/k}\right)$. 
\end{enumerate}

\end{theorem}

\begin{proof}
\emph{Correctness.}
If a $2$-path $P$ has an edge in $H$, Observation~\ref{obs:edge-in-H} applies.
Otherwise, by Lemma~\ref{lem:exists-sat} $P$ contains a saturated edge; if the other edge together with Phase~3/4 conditions satisfies Lemma~\ref{lem:key-lemma}, the desired bound holds already.
Any remaining bad $2$-path is fixed in Phase~5 by inserting its two edges.

The result for the paths of length one follows from the fact that $H$ contains a usual $(2k-1)$-spanner and for paths of length greater than $2$ it can be obtained by dividing the path into segments of lengths one and two, with at most one piece having length one.
\smallskip

\noindent\emph{Size and time.}
Lemmas/claims~\ref{lem:phase2}, \ref{clm:phase3}, \ref{clm:phase4}, and~\ref{clm:phase5} give
\[
  \card{E(H)} \;=\; O\!\left(n^{1+1/k}\right).
\]
All phases scan edges/paths and maintain weighted balls in subgraphs can be maintained and computed; shortest paths can be computed by Dijkstra's algorithm, hence the algorithm may be implemented in polynomial time. 
\end{proof}

\section{Fault Tolerant Spanners}\label{sec:ft-spanners}

\begin{definition}[EFT $(d\to r)$ spanner]\label{def:eft-vft}

Let $\G=(\V,\E)$ be a (multi)graph and $f\in\mathbb{N}$. A subgraph $H\subseteq \G$ is an \emph{$f$-edge-fault tolerant} ($f$-EFT) \emph{$(d\to r)$ spanner} if for every $F\subseteq \E$ with $|F|\le f$ and all $x,y\in \V$,
    \[
        \dist_{G-F}(x,y)=d \;\Longrightarrow\; \distX{H-F}(x,y)\le r.
    \]
    
\end{definition}

We next describe the greedy construction used throughout this section. 

\begin{algorithm}[h]
\caption{\textsc{EFT-Greedy $d\to r$ Spanner}$(\G,d,r,f)$}\label{alg:ft-greedy}
\DontPrintSemicolon

$H\gets(\V,\emptyset)$\;

\ForEach{ $x,y\in V$}{
  \ForEach{Path $P$ of length $d$ from $x$ to $y$}{}
  \If{$\exists\,F\subseteq \E$ with $|F|\le f$ and $F\cap E(P)=\emptyset$ s.t.\ $\dist_{H-F}\big(x,y\big)>r$}{
    $E(H)\gets E(H)\cup E(P)$\;
  }
}
\Return $H$\;
\end{algorithm}

By construction, the output of \Cref{alg:ft-greedy} is an $f$-EFT $(d\to r)$ spanner. To bound the size, we formalize the obstruction to adding a path via \emph{blocking sets}.

\begin{definition}[$(f,d\to r)$ blocking set]\label{def:blocking}
Let $P_1,\ldots,P_t$ be paths of length $d$ on a common vertex set $\V$. We say they admit an \emph{$(f,d\to r)$ edge-blocking set} if for every $i$ there exists a set $F_i\subseteq \E$ with $|F_i|\le f$ and $F_i\cap E(P_i)=\emptyset$ such that
\[
  \distX{\left(\bigcup_{j<i} P_j\right)-F_i}\!\big(x(P_i),y(P_i)\big) \;>\; r.
\]

Every ordered pair $(P_i,a)$ with $a\in F_i$ is called a \emph{block}.
\end{definition}

\begin{claim}
    Let $P_1, ..., P_t$ be the $d$-paths added by \Cref{alg:ft-greedy}. Then they admit an $(f, d\to r)$ edge-blocking set. 
\end{claim}
\begin{proof}
    For each $i$ let $F_i$ be the set $F$ required by \Cref{alg:ft-greedy} to add $P_i$ to the current spanner $H_i = \bigcup_{j < i} P_i$. It means that $F_i \cap E(P_i) = \emptyset$, $|F_i| \leq f$ and $\dist_{H_i - F_i}(x(P_i),y(P_i)) > r$. This shows that the sets $F_i$ satisfy \Cref{def:blocking}.
\end{proof}

In light of the above, it's enough to obtain combinatorial bounds on the size of path collections with an $(f,d\to r)$ blocking set, to conclude size bounds on the output of the FT greedy algorithm.
We first observe that one can reduce the size of FT greedy spanners and the corresponding standard greedy spanners by sub-sampling, as appears in  \cite{BodwinP19}. Let $g(n,d\to r)$ denote the maximum number of length-$d$ paths that can be added by the \emph{non-FT} greedy $(d\to r)$ process for \emph{path collections} \Cref{alg:greedy-spanner-path-collections} on an $n$-vertex input (i.e., the standard greedy that adds a path whenever the current distance between its endpoints exceeds $r$).

\begin{theorem}\label{thm:edge-blocking}
Suppose $P_1,\ldots,P_t$ admit an $(f,d\to r)$ edge-blocking set. Then
\[
  t \;=\; O\!\big(d\, f^{\,d}\, g(n,d\to r)\big).
\]
\end{theorem}

\begin{proof}
Sample each edge of $\bigcup_{i=1}^t E(P_i)$ independently with probability $p=c/f$ for a constant $c\in(0,1)$. Let $G_p$ be the sampled subgraph. A fixed path $P_i$ of length $j$ survives with probability $p^{j}$; a fixed block $(P_i,e)$ survives with probability $p^{j+1}$.

Let $X$ be the number of surviving paths and $Y$ the number of surviving blocks. Then
\[
\mathbb{E}[X]=t\,p^{d},\qquad \mathbb{E}[Y]\le t\,f\,p^{d+1}.
\]
Hence there exists a choice of $G_p$ in which the number $Z$ of surviving paths that do not participate in any surviving block satisfies
\[
  Z \;\ge\; X-Y \;\ge\; t\,p^{d}\bigl(1-fp\bigr) \;=\; t\,\Bigl(\tfrac{c}{f}\Bigr)^{\!d}\,(1-c).
\]
Choosing $c=\frac{d}{d+1}$ yields a subfamily of $s=\Omega\!\big(t/(d\,f^{d})\big)$ \emph{unblocked} surviving paths. Order these as $P_{b_1},\ldots,P_{b_s}$ by increasing original indices.

For each $\ell$, since $F_{b_\ell}$ blocks $P_{b_\ell}$ and contains no edge of any earlier $P_{b_j}$ (otherwise $(P_{b_\ell},e)$ would be a surviving block), we have
\begin{equation}\label{eq:addable}
  r \;<\; \distX{\left(\bigcup_{j<b_\ell} P_j\right)\!-\!F_{b_\ell}}\!\big(x(P_{b_\ell}),y(P_{b_\ell})\big)
  \;\le\; \distX{\bigcup_{j<\ell} P_{b_j}}\!\big(x(P_{b_\ell}),y(P_{b_\ell})\big).
\end{equation}

Thus $P_{b_\ell}$ is addable by the standard (non-FT) greedy $(d\to r)$ process for paths on the union graph $G'=(\V,\bigcup_{j=1}^{s}E(P_{b_j}))$. Consequently $s\le g(n,d\to r)$, and rearranging gives $t=O\!\big(d\, f^{d}\, g(n,d\to r)\big)$.
\end{proof}

\begin{remark}\label{rmk:ft-path-collection}
We used the fact that $g(n,d\to r)$ is a bound on the path collection greedy spanner, in the last paragraph in the previous proofs. In particular, it implies we do not need the selected paths $P_{b_i}$ to be consistent with the exact distance $d$ path structure of a single graph $G$.
\end{remark}

If one combines the last claim with \Cref{thm:main} one gets an $O(f^2n^{1+\frac1k})$ edges. Crucially, we now show one can improve on this vanilla sampling approach to obtain better size bounds for edge fault tolerant spanners.

\begin{theorem}\label{thm:comb-eft-2}
    Let $P_1, ..., P_t$ be the collection of paths that has a $(f,2 \to 2k)$ edge blocking set. Then $t = O(fn^{1+\frac{1}{k}})$. 
\end{theorem}

\begin{proof}
    Let $x_i, y_i$ be the endpoints of $P_i$ and $m_i$ its middle point. There are two possibilities: either $P_i \cap \bigcup_{j<i}P_j = \emptyset$ or this intersection is an edge. 
    
    We will first prove that the second option occurs at most $O(fn^{1+\frac{1}{k}})$ times. Let $P_{i_1}, ..., P_{i_s}$ be the paths ($i_1 < i_2 < ... < i_s$) for which the second option is in place. Let $e_{i_\ell} = P_{i_\ell} - \bigcup_{j < i_\ell}P_j$.
    Assume without loss of generality that $e_{i_\ell} = \{m_{i_\ell}, y_{i_\ell}\}$. By the definition of having a blocking set, there exists the set $F_{i_\ell} \subset E(\bigcup_{j < i_\ell}P_j)$, such that $|F_{i_\ell}| \leq f$, $F\cap E(P_{i_\ell}) = \emptyset$ and $\dist_{\bigcup_{j<i_\ell}P_j - F_{i_\ell}}(x_i, y_i) > 2k$. Since by definition of a blocking set $\{x_{i_\ell}, m_{i_\ell}\} \notin F_{i_{\ell}}$, $\dist_{\bigcup_{j<i_\ell}P_j - F_{i_\ell}}(x_i, y_i) \leq \dist_{\bigcup_{j<i_\ell}P_j - F_{i_\ell}}(m_i, y_i) + 1$. This implies that 
    $$\dist_{\bigcup_{j<\ell}e_{i_j} - F_{i_\ell}}(m_i, y_i) \geq \dist_{\bigcup_{j<i_\ell}P_j - F_{i_\ell}}(m_i, y_i) > 2k - 1.$$

    This implies that $e_{i_1}, ..., e_{i_s}$ can be sequentially added by the greedy algorithm for the standard EFT spanner for multigraphs with the stretch factor $2k-1$. It is known that the standard EFT spanner for multigraphs with the stretch factor $2k-1$ has size $O(fn^{1+\frac{1}{k}})$ which follows by the method of \cite{BodwinP19}, with a full proof appearing in \cite{PetruschkaST24_ITCS} as the first bullet of Theorem 3.1.

    Now we can assume that all $P_i$ are edge disjoint. Consider the random subgraph $H_p$ of $\bigcup P_i$, where for each $i$ \emph{both} edges of $P_i$ are in $H_p$ with probability $p$ and \emph{both}
    not in $H_p$ with probability $1-p$ independently for distinct paths. We note that this coupling of the edges of the same path is what distinguishes this proof from the approach used before and allows to achieve a better bound. Let $X$ be the set of surviving paths (i.e. those edges are in $H_p$), $Y$ be the set of surviving blocks and $Z$ be the set of surviving paths that do not participate in any surviving block. Then $\mathbb{E}[|X|] = pt$, $\mathbb{E}[Y] \leq ftp^2$ and thus, since $|Z| \geq |X| - |Y|$, $\mathbb{E}[Z] \geq pt - p^2 ft$, and so for at least one choice of the subgraph $H_p$ the set $Z$ has at least such size. Let us fix this choice. Let $Z = \{P_{i_1}, ..., P_{i_{|Z|}}\}$. Then for each $i_\ell$
    $$2k < \dist_{\bigcup_{j<i_\ell}P_j-F_{i_\ell}}(x_{i_\ell}, y_{i_\ell}) \leq \dist_{\bigcup_{a<\ell}P_{i_a}-F_{i_\ell}}(x_{i_\ell}, y_{i_\ell}) = \dist_{\bigcup_{a<\ell}P_{i_a}}(x_{i_\ell}, y_{i_\ell}),$$
    where the last equality is in place because $\bigcup_{a<\ell}P_{i_a}\cap F_{i_\ell} = \emptyset$, otherwise $P_{i_\ell}$ participates in surviving block. 

    This implies, by Theorem \ref{thm:combinatorial}, that $|Z| = O( n^{1+\frac{1}{k}})$. Consequently, $t = O(\frac{ n^{1+\frac{1}{k}}}{p - fp^2})$. Choosing $p = \frac{1}{2f}$, we get the desired bound. 
\end{proof}

\begin{observation}
    The union of an $f$-EFT $2\to 2k$ spanner and an $f$-EFT $1 \to 2k-1$ spanner is a $(k, k-1)$ $f$-EFT spanner. In particular, if we take the $2\to 2k$ spanner constructed above and an optimal $1\to 2k-1$ spanner, we get a $(k, k-1)$ $f$-EFT spanner of size $O(kfn^{1+\frac{1}{k}})$.
\end{observation}
\begin{proof}
    Let $H^1$ be the $f$-EFT $1 \to 2k-1$ spanner, $H^2$ be the $f$-EFT $2 \to 2k$ spanner and $H := H^1 \cup H^2$. Fix a set of faults $F$. Suppose that for some vertices $x$ and $y$ we have $\dist_{G-F}(x,y) = d$. Take a shortest path $(x=v_0, v_1, ..., v_d=y)$ between them. Then for each $0 \leq i \leq d-2$ we have $\dist_{G-F}(v_i, v_{i+2}) =2$ and, consequently, $\dist_{H-F}(v_i, v_{i+2}) \leq 2k$. Similarly, for each $0 \leq i \leq d-1$ we have $\dist_{G-F}(v_i, v_{i+1}) = 1$ and, consequently, $\dist_{H-F}(v_i, v_{i+1}) \leq \dist_{H^1-F}(v_i, v_{i+1}) \leq 2k-1$. Then, in case $d$ is even we get $\dist_{H-F}(x,y) \leq \sum_{i = 0}^{\frac{d}{2}-1}\dist_{H-F}(v_{2i},v_{2i+2}) \leq \frac{d}{2}\cdot 2k = kd$. In case $d$ is odd we get $\dist_{H-F}(x,y) \leq \sum_{i = 0}^{\frac{d-3}{2}}\dist_{H-F}(v_{2i},v_{2i+2}) + \dist_{H-F}(v_{d-1}, v_{d}) \leq \frac{d-1}{2}\cdot 2k + (2k - 1) = kd + (k-1)$.
\end{proof}

\paragraph{A Matching Lower Bound} The size of the $f$-EFT $2\to 2k$ spanner we obtained here is the best possible assuming the girth conjecture of Erd\H{o}s. Take an Erd\H{o}s graph on $\frac{n}{2}$ vertices with $\Omega(n^{1+\frac{1}{k}})$ edges and girth strictly greater than $2k+1$, replace each edge by $f$ parallel edges and then attach a leaf to each vertex. Such a graph has no proper $f$-EFT $2\to 2k$ spanner. For $f$-EFT $(k,k-1)$-spanner the same construction without adding leaves has no proper $f$-EFT $(k,k-1)$-spanner.

\section{Conclusions}
In this work we suggested a simple greedy procedure to obtain a $d\to r$ spanner, and gave tight constructions of $f$-EFT $2 \to 2k$ spanner for multigraphs, as well as a construction that takes a graph $G$ as input, and outputs a subgraph $H$ approximating weighted paths of hop-length $2$ in $G$, with optimal size/stretch tradeoff.

In follow-up work together with Parter, we give an $O_k(fn^{1+\frac{1}{k}})$ bound for $2\to 2k$ spanners supporting vertex faults, nearly matching the lower bound of \cite{BodwinDPW18}. We also prove that the greedy $d\to r$ spanner has similar size guarantees to the $(k^{\varepsilon},O_{\varepsilon}(k))$ spanners of Ben-Levy and Parter \cite{Ben-LevyP20}, and analyze the corresponding FT-greedy spanner to obtain \emph{truly} polynomial (in $f$) constructions achieving the stretch of the constructions of \cite{Ben-LevyP20}, for every fixed $\varepsilon>0$.

\noindent We believe the following directions are intriguing for future study:

\begin{enumerate}

    \item We conjecture that EFT $(k,k-1)$-spanner for simple graphs have size at most $f^{\frac{1}{k}}$ bigger than the best possible bounds and in particular the bounds of \cite{BodwinDR22}. At the moment, it seems possible that the upper bounds of \cite{BodwinDR22} for $2k-1$ spanners apply to $(k,k-1)$-spanners, without requiring any overhead.
    
    \item Is it possible to produce light spanners achieving the weighted approximation in \Cref{thm:weighted-main}? The lightness of a spanner is the ratio between the total weight of its edges and the weight of a minimum spanning tree in $G$.

    \item Does \Cref{thm:FT-spanner-bounds-intro} extends to weighted FT spanners? We believe this is not merely a technical problem but a conceptual problem. The techniques presented in this paper are local (mostly handle paths of size $2,\sqrt{k}$) while weighted $(\alpha,\beta)$-spanners have a \emph{global} additive error. In particular, we would not be surprised if \Cref{thm:FT-spanner-bounds-intro} doesn't extend to $f$-FT $(k,k-1)$-spanners for weighted graphs, and find this possibility interesting.

    \item Can the stretch guarantee of \Cref{thm:weighted-main} be improved? What we mean by this question, in light of \Cref{clm:lower-bound-wtd}, is to give some bounds on the structure of paths in $G$ that are hard to approximate in the spanner. E.g. one possibility is it the case that for any $G=(V,E)$ there is a wtd $2\to2k$ spanner $H$ of size $\tilde{O}(n^{1+\frac{1}{k}})$, such that any $2$-path $P$ in $G$ with $E(P) \cap E(H) =\emptyset$ has stretch exactly $k$ in $H$?

\end{enumerate}

\bibliographystyle{alpha}
\bibliography{refs}

\appendix

\section{Polynomial Time, $f$-EFT $2\to 2k$ Spanner of Near Optimal Size}\label{sec:poly-time}

We note here that one may use the method of \cite{DinitzR20} to get an EFT spanner algorithm with polynomial running time and essentially same size guarantee (up to a factor $k$).

\begin{theorem}
    There is a polynomial (in $n,f$) time algorithm that constructs an $f$-EFT spanner of a given $n$-vertex graph $G$ with $O(kfn^{1+\frac{1}{k}})$ many edges. 
\end{theorem}

\begin{proof}

\begin{algorithm}[h]
\caption{$\mathsf{EFTModifiedGreedySpanner}(G, f)$}\label{alg:FT-greedy-modified-spanner}

    $H \gets (V, \emptyset)$
    
    \ForEach{$\text{2-path } P = (x,m,y) \in G$}{
        $j\gets 0$
        
        $H' \gets H\backslash E(P)$
        
        \ForEach{$i = 1, ..., f+1$}{
            $j \gets j+1$
            
            Find a path $P'_j(P)$ of length $\leq 2k$ with endpoints $x$ and $y$ in $H'$, if there is no such, end the for-loop
            
            $H' \gets H'\backslash E(P'_j(P))$
        }
        \If{$j < f+1$}{
        $E(H) \gets E(H) \cup E(P)$.
        }
    }
    \Return $H$.

\end{algorithm}

We claim that the algorithm \ref{alg:FT-greedy-modified-spanner} satisfies the conditions stated in the theorem. It is indeed polynomial time: One can go over $O(n^3)$ 2-paths in $G$, and for each such path run BFS $O(f)$ times. Hence the running time is $O(f \cdot poly(n))$.

Let $P=(x,m,y)$ be a path in $G$ which is not in $H$. Then by the algorithm, there exist $f+1$ edge-disjoint paths $P_1',\ldots,P_{f+1}'$ in $H$ from $x$ to $y$. By their disjointness $\forall F \subseteq E$, $|F| \leq f$, there exist $i$ with $F_i \cap P'_i = \emptyset$ hence for every such $F$, $\dist_{H-F}(x,y)\leq 2k$.

By Theorem \ref{thm:comb-eft-2}, it remains to show that the collection of paths added by the algorithm has a $(2kf, 2\to 2k)$ blocking set. This is indeed the case: for each path $P_i$ take $F_i = \bigcup_j P'_j(P_i)$. Any other path from $x_i$ to $y_i$ of length at most $2k$ disjoint from $P_i$ should contain an edge from $F_i$, as otherwise it would be added to the collection $\{P'_j(P_i)\}_j$ by the algorithm, and by the condition of adding $P_i$ the collection $\{P'_j(P_i)\}_j$ has at most $f$ paths and so $|F_i| \leq 2kf$. 

\end{proof}

\section{A $\sqrt{k} \to O(k)$ Greedy Spanner}\label{sec:greedy-sqrt-spanner}

\begin{theorem}
    For an $n$-vertex graph $G$ the greedy algorithm for the $\ceil{\sqrt{k}} \to C k$ spanner for some constant $C$ outputs a spanner with $O(kn^{1+\frac{1}{k}})$ edges.
\end{theorem}

This theorem is an easy corollary of the following result considering an $(O(\sqrt{k}), O(k))$-spanners, that we are going to obtain combining our techniques with the ideas in \cite{Ben-LevyP20}.

\begin{theorem}
    For any $k \geq 1$, $d\geq 1$ the greedy algorithm constructs a subgraph $H$ of a (unweighted) graph $G$ on $n$ vertices such that for any path $P$ of length $d$ with endpoints $u$ and $v$ we have $\dist_{H}(x,y) \leq 4\ceil{\sqrt{k}}^2 + 2(2\ceil{\sqrt{k}} - 1)d$ and $|E(H)| = O(d\sqrt{k}\,n^{1+\frac{1}{k}})$.
\end{theorem}

\begin{proof}
    Denote $R:= \ceil{\sqrt{k}}$. We call the vertex $v$ fully clustered on some stage of the algorithm if for each $\ell = 0, 1, ... R$ we have $|B_H(v,\ell)| \geq n^{\frac{\ell}{k}}$.

    Let us call a path $P$ that is added by the greedy algorithm a clustering path if there is an edge $e = \{u,v\} \in P$ such that either $u$ or $v$ is not fully clustered and $\dist_H(u,v) > 2R - 1$ before adding $P$.
    
    \begin{claim}\label{cl:num-cl-paths-greedy}
        There are $O(\sqrt{k}\,n^{1 +\frac{1}{k}})$ clustering paths.
    \end{claim}
    \begin{proof}
        For a clustering path $P$ let $e = {u,v}$ be as in the definition. Let $\ell$ be the maximal integer such that $|B_H(u,\ell'-1)| \geq n^{\frac{\ell'-1}{k}}$ and $|B_H(v,\ell'-1)| \geq n^{\frac{\ell'-1}{k}}$ for all $\ell' \in [\ell]$. By the choice of $e$, $\ell \leq R$. Let w.l.o.g. $B_H(u,\ell) < n^{\frac{\ell}{k}}$. Note that $B_H(u,\ell)\cap B_H(v,\ell-1) = \emptyset$. Indeed, suppose there is an $x\in B_H(u,\ell)\cap B_H(v,\ell-1)$. Then $\dist_{H}(u,v) \leq \dist_H(u,x) + \dist_H(x,v) \leq 2\ell - 1 \leq 2R-1$, contradiction. On the other hand, by triangle inequality, $B_H(v,\ell-1) \subset B_{H\cup \{e\}}(u, \ell)$. Overall, we have that $B_H(v,\ell-1) \subset B_{H\cup \{e\}}(u, \ell) - B_{H}(u, \ell)$, and so $|B_{H\cup \{e\}}(u, \ell)| - |B_{H}(u, \ell)| \geq n^{\frac{\ell - 1}{k}}$, while $|B_{H}(u, \ell)| < n^{\frac{\ell}{k}}$. This implies, that each vertex can play the role of $u$ at most $n^{\frac{1}{k}}$ times for each $\ell\in [R]$ during the algorithm. The claim follows. 
    \end{proof}

    \begin{claim}\label{cl:two-vertices-on-nc-path}
        For each path $P$ that is added by the greedy algorithm, but is not clustering the following holds (in the moment of adding):
        \begin{itemize}
            \item For each $e = \{u,v\} \in P$ if $u$ or $v$ is not fully clustered, then $\dist_H(u,v) \leq 2R + 1$;
            \item $P$ has at least two fully clustered vertices.
        \end{itemize}
    \end{claim}
    \begin{proof}
        The first part is obvious from the definition of clustering path. For the second part notice that if $P = (x = v_0, v_1, ..., v_{d-1}, v_d = y)$ has at most one fully clustered vertex, then each of its edges has a not fully clustered end and thus the distance in current $H$ between its ends is by the first statement at most $2R-1$. Consequently, $\dist_H(x,y) \leq \sum_{i=0}^{d-1}\dist_H(v_i, v_{i+1}) \leq d(2R - 1)$, what means that the path already has the right stretch and thus is not added by the algorithm. 
    \end{proof}

    Now we shall bound the number of added paths that are not clustering. To do so let us first select a set of distinguished vertices called centers (similarly to the Baswana-Sen approach), such that each fully clustered vertex has a center within distance $R$. 
    \begin{claim}\label{cl:set-of-centers}
        Let $C$ be a $p$-random subset of $V$, where $p = 10\ln n \cdot n^{-\frac{R}{k}}$ (assume $n$ is large enough for this to be less than $1$). Then both following conditions hold simultaneously with positive probability:
        \begin{itemize}
            \item $|C| \leq 20\ln n\cdot n^{1-\frac{R}{k}}$;
            \item For each fully clustered vertex $v$ there exists some $c\in C$ such that $\dist_{H}(v, c)\leq R$, where $H$ is the current subgraph immediately after $v$ becomes fully clustered (and thus the same holds in any moment afterwards). 
        \end{itemize}
    \end{claim}
    \begin{proof}
        Since $\mathbb{E}|C| = pn = 10\ln n\cdot n^{1-\frac{R}{k}}$, the first condition holds with probability at least $\frac{1}{2}$ by Markov's inequality. 

        By definition, each fully clustered vertex has at least $b=n^{R/k}$ vertices at distance at most $R$. Thus the probability that none of them is selected is at most $(1-p)^{b} \leq e^{-pb} = e^{-10\ln n} = n^{-10}$. Then the probability that the second condition is broken for at least on vertex is, by the union bound, at most $n^{-9}$, that is less than $\frac{1}{2}$ for any $n \geq 2$.
    \end{proof}

    Let us now fix one particular $C$ that satisfies both the above conditions. For each clustered vertex $v$ choose arbitrarily one center at distance $\leq R$ immediately after becoming fully clustered, and denote it $c_v$. Define the graph $\hat G$ with the vertex set $C$ and the edges $\{\{c_1, c_2\} : \dist_G(c_1, c_2) \leq 2R + d\}$. Define its subgraph $\hat H$, that is initially and empty graph on the same vertices and changes during the greedy algorithm on $G$ as described below. 

    By the Claim~\ref{cl:two-vertices-on-nc-path}, each path $P$ added by the algorithm, that is not clustering, has at least two clustered vertices. Denote $\ell(P)$ and $r(P)$ the first and the last (``leftmost'' and ``rightmost'') of them (here we think of $P$ as oriented from $x$ to $y$). We will write $\ell$ instead of $\ell(P)$ and $r$ instead of $r(P)$ for brevity. Each time a non-clustering path $P$ is added to $H$, we simultaneously add the edge $\{c_{\ell}, c_{r}\}$ to $\hat H$ (see Fig.~\ref{fig:long_path}). We have the following 

    \begin{claim}
        For step as described above
        \begin{itemize}
            \item $\{c_{\ell}, c_{r}\} \in E(\hat G)$;
            \item $\dist_{\hat H}(c_{\ell}, c_{r}) > 2R - 1$. In particular, our procedure does not try to add loops or double edges.
        \end{itemize}
    \end{claim}

    \begin{figure}[h]
    \centering    \includegraphics[width=0.7\linewidth]{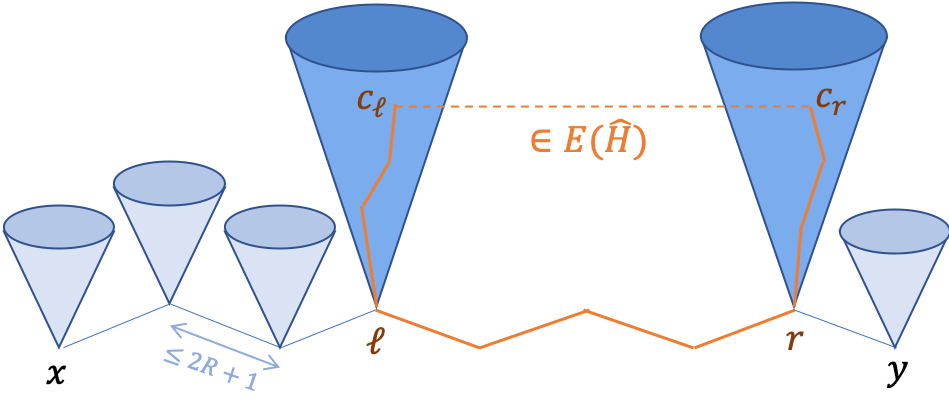}
    \caption{An addition of an edge to a virtual graph $\hat{H}$}
    \label{fig:long_path}
\end{figure}
    
    \begin{proof}
        
        The first claim is obvious: $\dist_G(c_\ell, c_r) \leq \dist_H(c_\ell, \ell) + \dist_P(\ell, r) + \dist_H(r, c_r) \leq 2R + d$. This argument actually implies a stronger statement that we shall use for the proof of the second part of the claim: if two centers $c_1$ and $c_2$ are connected by an edge in $\hat H$ on the current stage of the greedy algorithm, then in the current $H$ they are close: $\dist_H(c_1, c_2) \leq 2R + d$. We have this since once the edge $\{c_\ell, c_r\}$ is added, the path $c_\ell \rightarrow \ell \rightarrow r \rightarrow c_r$ that we measured above is entirely in $H$. 

        Assume the contrary to the second part of the claim: $\dist_{\hat H}(c_{\ell}, c_{r}) \leq 2R - 1$. Then, what we have just proven implies that $\dist_H(c_\ell, c_r) \leq (2R - 1)(d + 2R)$. Since all the edges in $P$ that come before $\ell$ or after $r$ have a non-clustered vertex and $P$ is not a clustering path, we know that the distance in $H$ between the ends of each of these edges is at most $2R - 1$. Hence, if we denote the endpoints of $P$ by $x$ and $y$, we have $\dist_H(x, \ell) \leq (2R-1) \dist_P(x, \ell)$ and $\dist_H(r, y) \leq (2R-1) \dist_P(r, y)$. Overall, we get 
        \begin{multline*}
            \dist_H(x,y) \leq \dist_H(x,\ell) + \dist_H(\ell, c_\ell) + \dist_H(c_\ell, c_r) + \dist_H(c_r, r) + \dist_H(r, y) \leq\\\leq (2R-1)d + 2R + (2R - 1)(d + 2R) =  2(2R - 1)d + 4R^2,
        \end{multline*}
        and so the path already has the right stretch and need not be added, contradiction.
    \end{proof}

    Note that, the preceding claim shows that the construction of $\hat H$ coincides with several first steps of the greedy algorithm that constructs a (usual) $2R-1$ spanner of $\hat{G}$. By a well-known girth argument, such an algorithm has at most $O(|V(\hat G)|^{1+ \frac{1}{R}})$ steps. Since $|V(\hat G)| = |C| = O(\ln n \cdot n^{1 - \frac{R}{k}})$, we have $|V(\hat G)|^{1+ \frac{1}{R}} = O(\ln^2 n \cdot n^{1 - \frac{1}{k}})$. So this quantity upper-bounds the number of added paths that are not clustering. 
    
\end{proof}

\end{document}